\documentclass[11pt,a4paper]{article}
\usepackage[margin=1.05in]{geometry}
\usepackage[T1]{fontenc}
\usepackage{amsmath,amssymb,amsthm}
\usepackage{booktabs}
\usepackage{microtype}
\usepackage{tikz}
\usepackage[hidelinks]{hyperref}
\hypersetup{pdftitle={Certified local rank and uniqueness barriers for a 48-term matrix-multiplication decomposition},pdfauthor={Abhinav Agarwal}}

\newcommand{\Kcert}{8}     % pointwise border radius
\newcommand{\KcertS}{11}   % pointwise strong radius
\newcommand{\KcertLP}{9}   % ambient generic strong/rank route
\newcommand{\KcertSG}{11}  % generic strong radius on the parameter curve
\newcommand{\KcertR}{12}   % pointwise rank radius
\newcommand{\KcertRG}{11}  % generic rank radius in the ambient component
\newcommand{\KcertRC}{12}  % generic rank radius on the actual curve only
\newcommand{\KcertSNext}{\number\numexpr\KcertS+1\relax}
\newcommand{\Knext}{\number\numexpr\KcertR+1\relax}

\newtheorem{theorem}{Theorem}[section]
\newtheorem{proposition}[theorem]{Proposition}
\newtheorem{lemma}[theorem]{Lemma}
\newtheorem{corollary}[theorem]{Corollary}
\theoremstyle{definition}
\newtheorem{definition}[theorem]{Definition}
\newtheorem{remark}[theorem]{Remark}

\DeclareMathOperator{\rank}{rank}
\DeclareMathOperator{\adj}{adj}
\DeclareMathOperator{\spn}{span}

\newcommand{\brank}{\underline{\mathrm{R}}}
\newcommand{\M}[1]{\langle #1\rangle}
\newcommand{\bbK}{\mathbb{K}}
\newcommand{\bbQ}{\mathbb{Q}}
\newcommand{\bbZ}{\mathbb{Z}}
\newcommand{\bbC}{\mathbb{C}}

\newcommand{\rrho}{\rho}
\newcommand{\brho}{\bar\rho}
\newcommand{\srho}{\rho_{\mathrm{strong}}}

\title{Certified local rank and uniqueness barriers\\
for a 48-term matrix-multiplication decomposition}
\author{Abhinav Agarwal\thanks{Independent researcher. Email: \href{mailto:abhinavagarwal1996@gmail.com}{\texttt{abhinavagarwal1996@gmail.com}}.\\Copyright \copyright\ 2026 Abhinav Agarwal. This article is licensed under \href{https://creativecommons.org/licenses/by/4.0/}{CC BY 4.0}.}}
\date{September 2026}

\begin{document}
\maketitle
\begin{abstract}
We study replacements in fixed bilinear tensor decompositions, counting
changes to complete rank-one summands, including output factors. The shortening
frontier records the maximum rank defect of a fixed-size subset and determines
the minimum length attainable within a change budget. For the rational
$48$-term Li--Wang--Hu decomposition $D(2)$ of $4\times4$ matrix multiplication
over $\bbC$, we prove rank radius at least $12$, strong radius exactly $11$,
and border radius at least $8$. Every shorter complex decomposition therefore
changes at least thirteen original summands. An exact rational twelve-term
replacement attains the equal-length barrier. The proofs combine exhaustive
support reductions with saturated projected kernels and zero-corner completion
arguments controlling arbitrary minimal competitors. A reduced-incidence
argument transfers kernel certificates to tensor-space neighborhoods. A
Laurent normal form gives strong radius exactly $11$ for the sixteen-term core
at every nonzero complex parameter. On a nonempty Zariski-open subset of the
actual parameter curve, the rank radius is at least $12$, the strong radius
exactly $11$, and the border radius at least $8$. We also prove incomparability
of the full Kothari--Moitra--Wein sufficient criterion and the
Sylvester-equipped kernel criterion. These results describe local decomposition
structure rather than a new rank bound for full matrix multiplication.
\end{abstract}


\section{Introduction and results}\label{sec:introduction}

How many complete rank-one summands of a bilinear tensor decomposition must
change before it can be shortened?  We count the whole tensor summand,
including its output factor.  For a length-$r$ decomposition $D$, let
$\Delta_D(k)$ be the largest rank defect among its $k$-term subsets.  Then
$r-\Delta_D(k)$ is the shortest length attainable while retaining at least
$r-k$ original tensor summands.  This is an endpoint-overlap statement: it
does not count graph edges, nor does it claim a lower bound on the number of
input products if output recombination is redesigned.

For the specified 48-term decomposition $D(2)$ of $\M{4,4,4}$ over $\bbC$, our results are
\[
 \rrho(D(2))\ge\KcertR,\qquad \srho(D(2))=\KcertS,\qquad
 \brho(D(2))\ge\Kcert.
\]
Thus every shorter complex decomposition discards at least $\Knext$ complete
summands, while a distinct 48-term decomposition can be obtained by changing
exactly $\KcertSNext$ of them.  The replacement retains 36 complete summands;
it retains 40 first-two-mode input directions, so its input-product distance
is eight.  No optimum for that different metric is asserted.

The reusable ingredients are the endpoint frontier identity, zero-corner
completion for arbitrary minimal competitors, and stacked saturated projected
kernels for global paired-factor recovery. The later sections apply them to
$D(t)$, combining them with Lovitz--Petrov splitting and exact positive
minors. The supplement contains applications, auxiliary geometry, complete
replay detail, and alternative rational-function transfers.

A second matrix-multiplication application is the $49$-term Kronecker square
$E^{\otimes2}$ of Strassen's algorithm. The supplement proves
$\brho(E^{\otimes2})\ge8$ over $\bbC$: every $48$-term decomposition
therefore discards at least nine summands of $E^{\otimes2}$. The exceptional scalar
extensions are settled by an exact Borel-fixed border-apolarity certificate
for $\brank_{\bbC}(\M{2,2,2}\oplus\M{1,1,1})=8$.

In the flip/reduction graph of Kauers and Moosbauer \cite{KM}, flips preserve
length and reductions shorten a decomposition.  Pairwise separation excludes
these moves.  Adaptive searches may also use the length-increasing
plus-transitions of Arai, Ichikawa and Hukushima \cite{AIH}; pairwise separation
does not preclude them.  Our bounds concern changed summands at an endpoint,
not graph-path reachability; in particular they do not resolve
Kauers--Moosbauer Question~5.

\begin{table}[htbp]
\centering
\begin{tabular}{@{}lll@{}}
\toprule
Scope & Rank radius & Strong radius / border radius\\
\midrule
$D(2)$ & $\ge\KcertR$ & $=\KcertS$ / $\ge\Kcert$\\
actual parameter curve, generically & $\ge\KcertRC$ & $=\KcertSG$ / $\ge\Kcert$\\
each component through $D(2)$, generically & $\ge\KcertRG$ & strong $\ge\KcertLP$, border $\ge\Kcert$\\
sixteen-term core, every $t\ne0$ & --- & strong $=\KcertS$\\
\bottomrule
\end{tabular}
\caption{All radii are over $\bbC$. The ambient row concerns each irreducible component through $D(2)$.
Uniform core rigidity is not a uniform rank-twelve theorem for the full family.}
\label{tab:scope}
\end{table}

Pointwise statements concern $D(2)$; curve-generic statements hold on a
nonempty Zariski-open subset of the actual parameter curve; ambient-generic
statements hold on a dense open subset of each component through $D(2)$; and
uniform core statements concern only the sixteen-term set $X$. Table~\ref{tab:scope}
records the strongest final bound at each scope.

\begin{table}[htbp]
\centering
\begin{tabular}{@{}lll@{}}
\toprule
Certificate & Residual supports & Role\\
\midrule
factor/commutator/Young & finite census & border eight and base rank\\
zero-corner completion & 16 ten-face cases & minimal-competitor purity\\
saturated projected kernels & 16 ten-face cases & global paired-factor recovery\\
splitting, minors, completions & eleven/twelve faces & strong through $11$; rank through $12$\\
\bottomrule
\end{tabular}
\caption{Proof dependency map for the residual support classes.}
\label{tab:proof-map}
\end{table}

\section{The shortening frontier}\label{sec:local}

Let $A,B,C$ be finite-dimensional spaces over a field $\bbK$,
$T\in A\otimes B\otimes C$, and $D=(d_1,\ldots,d_r)$ an ordered
decomposition with $d_i=a_i\otimes b_i\otimes c_i\ne0$. Rank is over
$\bbK$; border rank is over its algebraic closure. For $S\subseteq[r]$, write
$T_S=\sum_{i\in S}d_i$.  For displayed decompositions of $T$, $d(D,D')$ is
the number of summands of $D$ not reused by $D'$ as complete tensors:
$r$ minus the size of a maximum matching of their summand multisets.
Thus this directed distance includes changes to output factors. Gauge means
factor rescaling $(a,b,c)\mapsto(\lambda a,\mu b,(\lambda\mu)^{-1}c)$.
A decomposition is pairwise separated if no two factor vectors are
proportional in any mode.

\begin{definition}\label{def:rho}
For $0\le k\le r$, set $\delta_D(S)=|S|-\rank(T_S)$ and
$\Delta_D(k)=\max_{|S|=k}\delta_D(S)$.  The rank radius $\rrho(D)$ is the
largest $k\in\{0,\ldots,r\}$ for which every subset of size at most $k$ has rank its
cardinality.  The border radius $\brho(D)$ is defined with border rank in
place of rank.  The strong radius $\srho(D)$ is the largest $k$ for which
every such subset has its displayed decomposition as its only decomposition
of length at most its cardinality over $\bbK$, up to permutation and gauge.
We call $D$ strongly $q$-locally rigid if and only if $\srho(D)\ge q$.
The border defect $\bar\delta_D$ and frontier $\bar\Delta_D$ replace rank by
border rank. Put $\tau_q(D)=\min\{k\in\{0,\ldots,r\}:\Delta_D(k)\ge q\}$, with
$\min\varnothing=\infty$. In particular $\srho(D)\le\rrho(D)$ and
$\brho(D)\le\rrho(D)$.
\end{definition}

\begin{lemma}\label{lem:defect}
\emph{(a)} If $S\subseteq S'$, then $\delta_D(S)\le\delta_D(S')$.
\emph{(b)} If $S\cap S'=\emptyset$,
then $\delta_D(S\cup S')\ge\delta_D(S)+\delta_D(S')$.
\end{lemma}
\begin{proof} Use $\rank(T_{S'})\le\rank(T_S)+|S'\setminus S|$ and subadditivity of tensor rank. \end{proof}

\begin{theorem}[Exact frontier]\label{thm:frontier}
For $0\le k\le r$,
\[
r-\Delta_D(k)=\min\{\,|D'|:D'\text{ decomposes }T,\ d(D,D')\le k\,\}.
\]
\end{theorem}
\begin{proof}
Replacing a maximizing block by a minimal decomposition attains the left side.
Conversely, the unmatched block of a competitor with $m\le k$ changes has a
decomposition of length $|D'|-r+m$, so $|D'|\ge r-\Delta_D(m)\ge r-\Delta_D(k)$.
\end{proof}

\begin{proposition}[Tensor products]\label{prop:productfrontier}
For decompositions $D,E$ over the same field, group corresponding tensor
modes to form $D\otimes E$. For $1\le k\le |D|$, $1\le\ell\le|E|$,
\[
\Delta_{D\otimes E}(k\ell)\ge
\ell\Delta_D(k)+k\Delta_E(\ell)-\Delta_D(k)\Delta_E(\ell).
\]
\end{proposition}
\begin{proof} Tensor minimizing decompositions of maximizing $k$- and $\ell$-blocks. \end{proof}

\begin{corollary}[Tensor powers]\label{cor:powerfrontier}
For $m\ge1$ and $1\le k\le r$, if $\Delta_D(k)=\delta$, then
$\Delta_{D^{\otimes m}}(k^m)\ge k^m-(k-\delta)^m$.
\end{corollary}
\begin{proof} Tensor $m$ minimum-length replacements of a maximizing block. \end{proof}
If $\delta>0$, the guaranteed saving fraction within that block is
$1-(1-\delta/k)^m$, tending exponentially to one. Relative to the whole
algorithm it is $(k/r)^m[1-(1-\delta/k)^m]$, tending to zero for $k<r$.
These are consequences of rank submultiplicativity, not tensor-power
rigidity lower bounds or a new exponent from a proper block.

\begin{theorem}[Increment bound]\label{thm:increment}
$\Delta_D(k+1)\le\Delta_D(k)+2$ for $0\le k<r$.
\end{theorem}
\begin{proof} Remove one term from a maximizing $(k+1)$-subset; adding one rank-one tensor can decrease rank by at most one. \end{proof}

\begin{corollary}\label{cor:increment}
$\Delta_D(k)\le2\max(0,k-\rrho(D))$.  Consequently a competitor changing at most $k$ summands has length at least $r-2\max(0,k-\rrho(D))$.
\end{corollary}
\begin{proof} Iterate the increment bound from $\Delta_D(\rrho(D))=0$ and apply Theorem~\ref{thm:frontier}. \end{proof}

\begin{theorem}[Exact distance]\label{thm:exactdist}
Put $d_<(D)=\min_{|D'|<r}d(D,D')$, where $D'$ ranges over decompositions
of $T$ and $\min\varnothing=\infty$. If $D$ is not rank-optimal, then
$d_<(D)=\rrho(D)+1$; otherwise $d_<(D)=\infty$.
\end{theorem}
\begin{proof} Apply Theorem~\ref{thm:frontier} at the first $k$ with $\Delta_D(k)>0$. \end{proof}

\begin{theorem}[Endpoint barrier]\label{thm:geodesic}
For any sequence of flips, reductions, and plus-transitions from $D$ to $D'$, with arbitrary intermediate
lengths, $d(D,D')\le\rrho(D)$ implies $|D'|\ge|D|$.
\end{theorem}
\begin{proof} This is the endpoint statement of Theorem~\ref{thm:exactdist}; intermediate lengths are irrelevant. \end{proof}

\begin{theorem}[Kruskal transfer]\label{thm:kruskal}
Order the three factor families so their Kruskal ranks satisfy
$2\le k_A\le k_B\le k_C$, where Kruskal rank is the largest size for
which every subfamily is independent. Then the decomposition is strongly
$q$-locally rigid for
\[
q\le\min\left(r,k_A+k_B-2,\left\lfloor\frac{k_A+k_B+k_C-2}{2}\right\rfloor\right).
\]
\end{theorem}
\begin{proof}
For a subset of size $s$, the three Kruskal ranks are at least
$\min(s,k_A)$, $\min(s,k_B)$, and $\min(s,k_C)$. For every $2\le s\le q$ their
sum is at least $2s+2$, so Kruskal's theorem applies; singletons are immediate.
\end{proof}

Under the same ordered-rank hypothesis, the border-rank criterion of
Blomenhofer and Lovitz \cite[Thm.~4.4]{BL} gives
\[
\brho(D)\ge\min\left(r,k_A+k_B-1,
\left\lfloor\frac{k_A+k_B+k_C-1}{2}\right\rfloor\right).
\]
Indeed, on each subset it suffices that twice its size is at most the
sum of its three Kruskal ranks minus one. This certifies border rank,
not uniqueness.

\begin{theorem}\label{thm:open}
Let $V_r(T)$ be the quasi-affine variety of ordered length-$r$
decompositions with nonzero factors over an algebraically closed field.
For $0\le k\le r$ and every integer $q$, the locus
$\{D\in V_r(T):\bar\Delta_D(k)\le q\}$ is Zariski open.
On any irreducible $Y\subseteq V_r(T)$, the whole border frontier attains
its pointwise minimum on a single dense open subset. In particular, a
border-radius bound holding at one point of $Y$ holds on a dense open subset.
\end{theorem}
\begin{proof}
For each support, $D\mapsto T_S$ is a morphism and
$\sigma_L=\{\brank\le L\}$ is a closed affine secant cone (empty for
$L<0$). The complement of the stated locus is the finite union of preimages
of $\sigma_{k-q-1}$. Intersect the finitely many nonempty open loci attaining
the minimum of each frontier value on the irreducible set $Y$.
\end{proof}

\section{Certificates}\label{sec:cert}

\begin{definition}
$D$ is \emph{Khatri--Rao nondegenerate} if in each of the three pairings the $r$ vectors $\{b_i\otimes c_i\}$,
$\{c_i\otimes a_i\}$, $\{a_i\otimes b_i\}$ are linearly independent.
\end{definition}

\begin{lemma}\label{lem:kr}
If $D$ is Khatri--Rao nondegenerate then for every $S$ and every mode, the rank of the mode-$m$ flattening of
$T_S$ equals the dimension of the span of the mode-$m$ factor vectors indexed by $S$.
\end{lemma}

\begin{proof}
Write the mode-$A$ flattening as $\sum_{i\in S}a_i\kappa_i^{\mathsf T}$, $\kappa_i=b_i\otimes c_i$. Its column
space lies in $\spn\{a_i\}_{i\in S}$. The $\kappa_i$, $i\in S$, are independent, so there are functionals
$\varphi_i$ with $\varphi_i(\kappa_j)=\delta_{ij}$; contracting returns each $a_i$.
\end{proof}

\begin{corollary}\label{cor:krcons}
Let $D$ be Khatri--Rao nondegenerate. Then no subset admits the \emph{regrouping rewrite}, which expresses one
summand's Khatri--Rao vector through those of the others and reabsorbs it; and $S$ with $|S|=k$ satisfies
$\brank(T_S)=k$ as soon as some mode contributes $k$ independent factor vectors.
\end{corollary}

\begin{proof}
A Khatri--Rao dependence within $S$ would give $\rank(T_S)\le|S|-1$ by regrouping, which nondegeneracy
excludes. The second claim is Lemma~\ref{lem:kr} with $\brank(T_S)\le|S|$.
\end{proof}

Throughout, a subset a certificate does not reach is one needing a further argument, never one shown to be
reducible. Nondegeneracy asks $r$ vectors to be independent in dimension $\dim B\cdot\dim C$, a genericity
condition rather than a property distinguishing a particular scheme; its role is to cut the cost per subset
from three eliminations on $\dim(B\otimes C)$ rows to three on $\dim A$ rows, collapsing the test to factor
rank. A Khatri--Rao dependence is an explicit reduction, so a scan certifies rigidity and enumerates candidate
regrouping sites in the displayed decomposition. This static test does not
answer \cite[Question~5]{KM}, which asks about reachable reduction edges.

\begin{lemma}[Subadditivity]\label{lem:subadd}
If every subset of size exactly $k$ satisfies $\brank(T_S)=k$, so does every subset of size $j\le k$ with $j$
in place of $k$. The same holds for rank.
\end{lemma}

\begin{proof}
If $|S|=j$ had $\brank(T_S)\le j-1$, subadditivity would give $\brank(T_{S'})\le(j-1)+(k-j)=k-1$ for any
superset $S'$ of size $k$.
\end{proof}

\begin{lemma}[Soundness modulo a prime]\label{lem:modp}
Let $f$ have integer coefficients and vanish identically on $\sigma_L$ over $\bbC$. If the entries of $T$ are
$p$-integral rationals for a prime $p$ and $f(T)\not\equiv0\pmod p$, then $\brank(T)>L$.
\end{lemma}

\begin{proof}
$\brank(T)\le L$ gives $f(T)=0$, an identity between $p$-integral rationals, which reduces modulo $p$.
\end{proof}

Every $(L+1)\times(L+1)$ flattening minor is such an $f$. Only positive verdicts lift: a rank deficiency modulo
$p$ is a flag, not a proof.

We use Strassen's inequality \cite{Strassen83}, in the form of \cite[Thm.~6.1.1]{LandsbergSurvey}: if
$T\in\bbK^m\otimes\bbK^n\otimes\bbK^n$ has first-mode slices $T_0,\dots,T_{m-1}$ and some combination $T_a$ is
invertible, then $\brank(T)\ge n+\frac12\rank[T_bT_a^{-1},T_cT_a^{-1}]$ for any further combinations $T_b,T_c$.

\begin{lemma}[Commutator certificate]\label{lem:comm}
Let $p$ be a prime, let $T$ have $p$-integral rational entries with first-mode slices of size $n\times n$, let
$T_a,T_b,T_c$ be integer combinations of them, and put $g=[T_b\adj(T_a),T_c\adj(T_a)]$. Let $L\ge n$. If modulo
$p$ both $\det T_a\not\equiv0$ and some $(2(L-n)+1)$-minor of $g$ is nonzero, then $\brank(T)>L$. Equivalently, a
nonzero $r$-minor of $g$ modulo $p$ gives $\brank(T)\ge n+\lceil r/2\rceil$; the case $r=1$, a single nonzero
entry, is $\brank(T)\ge n+1$.
\end{lemma}

\begin{proof}
Put $h=\det T_a$, so that $g$ and $h$ are integer polynomials in the entries of $T$ and
$g=h^2[T_bT_a^{-1},T_cT_a^{-1}]$ wherever $h\ne0$. On $\sigma_L\cap\{h\ne0\}$, Strassen's inequality
$\brank\ge n+\frac12\rank[T_bT_a^{-1},T_cT_a^{-1}]$ \cite{Strassen83} together with $\brank\le L$ forces
$\rank g\le2(L-n)$, so every $(2(L-n)+1)$-minor of $g$ vanishes there. Since $\sigma_L$ is irreducible, either
that open set is nonempty, hence dense, so every such minor vanishes on $\sigma_L$ and Lemma~\ref{lem:modp}
applies to one of them; or it is empty, so $h$ vanishes on $\sigma_L$ and Lemma~\ref{lem:modp} applies to $h$.
For the restatement, a nonzero $r$-minor certifies $\brank>L$ for the largest $L$ with $2(L-n)+1\le r$, namely
$L=n+\lfloor(r-1)/2\rfloor$, and $\lfloor(r-1)/2\rfloor+1=\lceil r/2\rceil$.
\end{proof}

The rank form is used by the $712$ level-$8$ subsets of \S\ref{sec:family}; the rank of $g$ is computed modulo $p$, and since reduction can only lower it,
a modular rank is a lower bound for the rank over $\bbQ$, which is the direction the inequality needs.

To apply this to a subset $S$, project each mode by an integer matrix to dimension $n$, chosen so that the
resulting slices admit an invertible combination; border rank does not increase under linear maps in each mode,
so a lower bound for the projected tensor is one for $T_S$. No multilinear support computation is needed, and
certificates record the projections.

One implementation detail affects the counts, not just efficiency. The lemma requires choosing which mode
supplies the slices, and when the projected dimensions coincide all three are admissible; one mode's slice
family can be identically singular while another's is not, so fixing the mode order silently reports artifacts.
Trying all admissible assignments changed the level-$6$ subsets of the Kronecker square certified by the
commutator bound from $14{,}642$ to all $14{,}930$.

A fourth certificate is the Young flattening bound of Landsberg and
Ottaviani \cite[Thm.~2.1]{LO15}, there called the skew-symmetrized flattening: for $T\in A\otimes B\otimes C$ with $\dim A=a$ and $0\le p\le a-1$, the induced map
$\Lambda^{p}A\otimes B^{*}\to\Lambda^{p+1}A\otimes C$ satisfies
$\brank(T)\ge\rank(T_A^{\wedge p})/\binom{a-1}{p}$. Like the first two it bounds border rank. Here too the
rank of $T_A^{\wedge p}$ may be computed modulo the prime, since reduction can only lower it and the bound
needs a lower bound for the rank over $\bbQ$.

When a partial sum has multilinear rank $(n,n,d)$ with $n=|S|-1$ and no
invertible third-mode slice combination, the next two lemmas treat this
singular-slice case.

\begin{lemma}[Singular pencil]\label{lem:pencil}
Let $\bbK$ be infinite, $|S|=k$, and suppose $T_S$ has multilinear rank $(n,n,d)$ with $n=k-1$. If every
combination of the third-mode slices is singular, then $\rank_{\bbK}(T_S)=k$.
\end{lemma}

\begin{lemma}\label{lem:cb}
In the situation of Lemma~\ref{lem:pencil}, let $A$ and $B$ be the $n\times k$ matrices of compressed factor
vectors in the two modes of rank $n$ and $c_1,\dots,c_k$ the third-mode vectors. Then
$\det(\sum_j\lambda_jS_j)=\sum_{|Q|=n}\det(A_Q)\det(B_Q)\prod_{q\in Q}\langle c_q,\lambda\rangle$, the sum over
$n$-element subsets $Q$ of the columns. If $\det(A_Q)\det(B_Q)=0$ for every such $Q$ the determinant is the
zero polynomial, and Lemma~\ref{lem:pencil} applies.
\end{lemma}

Proofs of these two elementary lemmas appear in Appendix~B of the mathematical
supplement.

Lemma~\ref{lem:pencil} is essentially the contrapositive of the classical fact that a concise tensor of minimal
rank is $1$-generic \cite{LandsbergBook}, and in the two-slice case the phenomenon goes back to Ja'Ja'
\cite{JaJa}; it certifies rank and not border rank, which is why we report the two radii separately.
Lemma~\ref{lem:cb} is the Cauchy--Binet formula. We enumerate the
$\binom{k}{n}$ coefficient products exhaustively and using their simultaneous vanishing as an exactly checkable
rank certificate, in $2\binom{k}{n}$ integer determinants of size $n$ per subset, with no inversion and no
fractions. These rational certificates remain valid after extension from $\bbQ$ to $\bbC$;
Lemmas~\ref{lem:kr} and \ref{lem:comm} bound border rank over an algebraic closure by construction.
Explicit small-block replacements can also be checked by exact tensor
identity; their use here does not require a complete rational rank oracle.

\section{The 48-multiplication family}\label{sec:family}

The factor spaces of $\M{4,4,4}$ have dimension $16$. Dumas, Pernet and Sedoglavic \cite{DPS1}
give the $48$-multiplication algorithm; Li, Wang and Hu \cite{LWH} give its parameterization by $48$ triples of
$4\times4$ matrices $(U_i(t),V_i(t),W_i(t))$ over $R=\bbZ[1/2,t,t^{-1}]$ with
\begin{equation}\label{eq:family}
\M{4,4,4}=\sum_{i=1}^{48}\operatorname{vec}U_i(t)\otimes\operatorname{vec}V_i(t)\otimes\operatorname{vec}W_i(t),
\end{equation}
the rational scheme. Write $D(t)$ for this decomposition, $D(\tau)$ for its
specialization. We verified \eqref{eq:family} on all $4096$ coordinates as an identity in $\bbQ(t)$; every
coefficient lies in $R$, of the form $(p+qt)/(2^at^k)$ with $a\in\{0,1,2,3\}$, $k\in\{0,1\}$. Vectors are
row-major flattenings of the stored matrices, with no further normalization.

\begin{theorem}\label{thm:sep}
For every pair $1\le i<j\le48$ and every mode there are coordinates $p<q$ whose $2\times2$ minor equals
$\pm2^at^k$, a unit of $R$.
\end{theorem}

\begin{proof}
Exhaustive verification in exact arithmetic in $R$ over all
$3\binom{48}{2}=3384$ pairs-and-modes. For each of them a witness minor is exhibited and its value
recorded; every recorded value is $\pm2^at^k$ with $a,k\in\bbZ$, a unit of $R$. The
exponents depend on the normalization of the stored vectors, being a unit does not. No case is settled by
absence: each of the $3384$ carries a witness.
\end{proof}

A single unit minor generates the unit ideal of maximal minors, which is
exactly what independence under every ring map $R\to\bbK$ requires. A unit \emph{gcd} of the minors would not
suffice: the minors $3$ and $t-1$ have gcd $1$ yet both vanish at $t=1$ in characteristic $3$.

\begin{corollary}\label{cor:allspec}
For any field $\bbK$ with $\operatorname{char}\bbK\ne2$ and any $\tau\in\bbK^{\times}$, $D(\tau)$ is a valid
$48$-term decomposition of $\M{4,4,4}$, is pairwise separated, and is an isolated vertex of the flip graph at
length $48$. With \cite{LWH}, the characteristic-zero flip graph at length $48$ has infinitely many pairwise
inequivalent isolated vertices.
\end{corollary}
\begin{proof}
The tensor identity specializes over $R$, and the unit minors of Theorem~\ref{thm:sep} give pairwise
separation, which excludes every flip at length $48$.
The invariant constructed in \cite[Section~5.2.1, pp.~22--23]{LWH} takes infinitely many values on
the nonzero rational parameters of this family, giving infinitely many distinct isotropy classes.
\end{proof}

A unit minor cannot vanish under any admissible specialization, so no exceptional nonzero parameter exists;
powers of $2$ cease to be units in characteristic $2$, which is where the certificate stops.

\begin{theorem}\label{thm:profile}
For every $\tau\ne0$ in every field of characteristic other than $2$: the first mode of $D(\tau)$ has $k$-rank
exactly $2$, with exactly $16$ dependent triples, and these partition $\{1,\dots,48\}$; the second and third
modes have $k$-rank at least $3$, every one of the $\binom{48}{3}=17296$ triples being independent in each.
Hence $\srho(D(\tau))\ge3$ and $D(\tau)$ is strongly $3$-locally rigid.
\end{theorem}

\begin{proof}
The $16$ dependent triples of the first mode were found by exact elimination over $\bbQ(t)$; each satisfies a
relation with integer coefficients independent of $t$, so it is dependent under every specialization, and the
$16$ are pairwise disjoint, hence partition $\{1,\dots,48\}$. For each of the remaining
$17296-16=17280$ triples of the first mode, and for each of the $17296$ triples of the second and of the third,
a $3\times3$ minor equal to a unit $\pm2^at^k$ of $R$ is exhibited, which by the remark after
Theorem~\ref{thm:sep} gives independence under every ring map $R\to\bbK$ with $\operatorname{char}\bbK\ne2$ and
$\tau\ne0$. The verification is exhaustive over all $3\cdot17296$ triples-and-modes: each is settled by an
exhibited witness, a dependence relation in $16$ cases and a unit minor in the rest, and none by absence.
Hence the first mode has $k$-rank exactly $2$, at most $2$ because a
dependent triple exists and at least $2$ by Theorem~\ref{thm:sep}, and the second and third have $k$-rank at
least $3$. Interchange the last two modes if necessary so that $k_B\le k_C$.
Theorem~\ref{thm:kruskal} applies, its hypothesis $k_A\ge2$ holding, and gives
\[
\srho(D(\tau))\ge\min\Bigl(48,\;k_B,\;\Bigl\lfloor\tfrac{k_B+k_C}{2}\Bigr\rfloor\Bigr)\ge3,
\]
since $k_A=2$ and $k_B,k_C\ge3$; so $D(\tau)$ is strongly $3$-locally rigid.
\end{proof}

We do not claim the second and third $k$-ranks equal $3$, which would need a dependent quadruple
at every parameter; only the lower bound is used, and Theorem~\ref{thm:kruskal} gives $\srho(D(\tau))\ge3$.

\begin{corollary}\label{cor:other48}
For every such $\tau$, any $48$-term decomposition of $\M{4,4,4}$ not equal to $D(\tau)$ up to permutation and
gauge differs from it in at least four summands.
\end{corollary}
\begin{proof}
Cancel a maximum matching of common summands: if a distinct $48$-term decomposition reused at least $45$,
the two residual decompositions would have the same length at most three, contradicting strong
$3$-local rigidity from Theorem~\ref{thm:profile}.
\end{proof}

\begin{corollary}\label{cor:universal3}
For every such $\tau$, $\brho(D(\tau))\ge3$.
\end{corollary}

\begin{proof}
The border transfer bound after Theorem~\ref{thm:kruskal} with profile $(2,3,3)$ gives $\min(48,4,3)=3$.
\end{proof}

\subsection{Structure}

Each dependent triple satisfies a relation with constant integer coefficients in $\{-2,-1,1,2\}$, independent
of $t$, such as $-U_1-U_{26}+2U_{28}=0$; in every triple two of the three matrices have rank one and the third,
of rank two, is that combination of them, here $U_{28}=(U_1+U_{26})/2$. The summands split into $32$ with all
three factor matrices of rank one and $16$ with all three of rank two, a split already in \cite{LWH}, whose set
$R_2$ is our $X=\{4,9,12,15,17,19,22,25,28,32,34,35,38,40,44,48\}$.

\begin{theorem}\label{thm:collapse}
For every $\tau\ne0$ in every field of characteristic other than $2$: $X$ is a
transversal of the $16$ dependent triples; its factor vectors span exactly $8$
dimensions in each of the three $16$-dimensional factor spaces, and all three
flattening ranks of $T_X$ equal $8$; and the complementary $32$ summands span
all $16$ dimensions in every mode.
\end{theorem}

\begin{proof}
The $16$ dependent triples are parameter-free: each holds with constant integer coefficients independent of
$t$, so the same $16$ index sets are the dependent triples at every $\tau\ne0$. $X$ is the explicit
$16$-element index set displayed above, and it meets each of the $16$ triples exactly once. This is a check on
two fixed families of index sets; no parameter and no matrix rank enters it.

In each mode the $16\times16$ matrix of factor vectors indexed by $X$ has rank $8$ over $\bbQ(t)$, by row
reduction in the fraction field, so all its $9\times9$ minors vanish identically in $R$ and no specialization
can raise the rank. In the same three modes the $8\times8$ minors
\[
  U:\;-2^{4},\qquad V:\;2,\qquad W:\;-2^{7}t
\]
are exhibited, each a unit of $R$; as in Theorem~\ref{thm:sep} a unit is carried to a unit, hence to a nonzero
element, by every ring map $R\to\bbK$ with $\operatorname{char}\bbK\ne2$ and $t\mapsto\tau\in\bbK^{\times}$,
so the rank is exactly $8$ at every such specialization. For the complementary $32$ summands the
corresponding $16\times16$ minors are $2^{17}$, $2^{6}t$ and $-2^{6}t$, again units of $R$, so those factor
vectors span at least $16$ dimensions and hence all $16$. Finally, by Theorem~\ref{thm:main-cert} $D(\tau)$ is
Khatri--Rao nondegenerate at every such $\tau$, so Lemma~\ref{lem:kr} equates each flattening rank of $T_X$
with the corresponding factor-span dimension, namely $8$.
\end{proof}

So every subset of $X$ of size at least $9$ is factor-dependent in all three modes, beyond the reach of any
flattening certificate. The collapse is peculiar to $X$: of the $2^{16}=65{,}536$ transversals of the $16$
dependent triples that are disjoint from $X$, some flattening certifies rank $16$ on $53{,}344$; on the
remaining $12{,}192$ no flattening does, and their rank is not determined here.

\subsection{Main results}

The results split into a universal tier, resting only on the profile, and a census at one parameter extended to
generic $t$ by openness; the gap is one of computational cost. The tiers differ in field as well as in level:
the universal tier holds in every characteristic other than $2$, whereas the census is a computation at
$\tau=2$ over $\bbQ$ and its conclusions are statements in characteristic zero. We do not claim the census
levels in positive characteristic, and the reason is not caution but the shape of the certificates. What buys
the universal tier its reach is that its witnesses are \emph{units} of $R$, which no admissible specialization
can kill; the census witnesses are nonzero rationals, and a nonzero rational can vanish modulo a prime. The
line between the tiers is therefore drawn by the certificate and not by the level.

\begin{theorem}[Universal]\label{thm:tier1}
For every $\tau\ne0$ in every field of characteristic other than $2$, $D(\tau)$ is pairwise separated, strongly
$3$-locally rigid, and satisfies $\brho(D(\tau))\ge3$.
\end{theorem}

\begin{proof}
Corollary~\ref{cor:allspec}, Theorem~\ref{thm:profile} and Corollary~\ref{cor:universal3}.
\end{proof}

\begin{theorem}\label{thm:main-cert}
$D(t)$ is Khatri--Rao nondegenerate over $\bbQ(t)$: in each pairing the $48$ Khatri--Rao vectors are independent
in the $256$-dimensional ambient space. Moreover, in each pairing some $48\times48$ minor of the Khatri--Rao
matrix is a \emph{unit} of $R=\bbZ[1/2,t,t^{-1}]$. Hence for every field $\bbK$ with
$\operatorname{char}\bbK\ne2$ and every $\tau\in\bbK^{\times}$, the specialization $D(\tau)$ is Khatri--Rao
nondegenerate.
\end{theorem}

\begin{proof}
Every coefficient of the family is of the form $(p+qt)/(2^at^k)$ with $a\in\{0,1,2,3\}$ and $k\in\{0,1\}$, so
multiplying each factor vector by a suitable $\pm2^{a}t^{k}$ --- a unit of $R$, which changes no linear
dependence --- puts all $3\cdot48$ factor vectors in $R^{16}$, and in fact in $\bbZ[t]^{16}$ with entries of
degree at most $1$. The Khatri--Rao vectors of a pairing are the $48$ coordinatewise products, so they lie in
$R^{256}$ with entries of degree at most $2$, and a $48\times48$ minor is a polynomial of degree at most $96$;
each minor below was computed exactly by evaluation at $97$ integer points and Lagrange interpolation. In the
three pairings the minors on the column sets recorded in the certificate file are
\[
  (V,W):\;-2^{52}t^{18},\qquad (U,W):\;-2^{49}t^{10},\qquad (U,V):\;+2^{49}t^{9},
\]
each a unit of $R$; the pairing written $\{c_i\otimes a_i\}$ in the definition is computed here as
$\{a_i\otimes c_i\}$, and the two differ by the coordinate transposition of $\bbK^{16}\otimes\bbK^{16}$, a
linear isomorphism and so immaterial for independence. A unit is carried to a unit, hence to a nonzero element,
by every ring homomorphism; and for $\bbK$ of characteristic $\ne2$ and $\tau\in\bbK^{\times}$ the assignment
$t\mapsto\tau$ defines a homomorphism $R\to\bbK$, since $2$ and $t$ are precisely the elements inverted in $R$.
So the minor does not vanish at $\tau$, the $48$ Khatri--Rao vectors of $D(\tau)$ are independent, and
$D(\tau)$ is Khatri--Rao nondegenerate.
\end{proof}

This is the argument of Theorem~\ref{thm:sep} applied to a $48\times48$ minor rather than a $2\times2$ one: a
unit of $R$ cannot vanish under any admissible specialization, so there is no exceptional parameter and no
exceptional characteristic other than $2$. A unit \emph{gcd} of the minors would again not suffice, for the
reason given after Theorem~\ref{thm:sep}. In particular Lemma~\ref{lem:kr} and Corollary~\ref{cor:krcons} are
available at every $\tau\ne0$, so a subset whose factor vectors are independent in some mode has
$\brank(T_S)=|S|$ there.

Which certificate applies determines which radius is bounded, so we name it at each use.

\begin{theorem}[Lovitz--Petrov transfer]\label{thm:lp}
Let $D$ be an ordered length-$r$ decomposition, so that in particular every summand $d_i$ is nonzero, and for
$U\subseteq\{1,\dots,r\}$ and each mode $m$ put $e_m(U)=|U|-\dim\spn\{\text{mode-$m$ factors indexed by }U\}$.
If
\[
  e_A(U)+e_B(U)+e_C(U)\;\le\;|U|-2\qquad\text{for every $U$ with }2\le|U|\le k,
\]
then $D$ is strongly $k$-locally rigid and $\rank(T_S)=|S|$ for every $|S|\le k$; hence $\Delta_D(k)=0$ and
$\rrho(D)\ge k$. The criterion bounds rank only and gives no information about $\brank$.
\end{theorem}

\begin{proof}
Fix $S$ with $|S|\le k$ and apply \cite[Thm.~2]{LovitzPetrov} to the $|S|$ product tensors $D_S$, which are
nonzero because the factors of an ordered decomposition are (Section~\ref{sec:local}), as that theorem requires
of its multiset; its hypothesis $2|U|\le\sum_m(d^U_m-1)+1$ for all $U\subseteq S$ with $2\le|U|\le|S|$ is the
displayed inequality after $d^U_m=|U|-e_m(U)$. Its conclusion is that $D_S$ is the unique tensor rank
decomposition of $T_S$, which in \cite{LovitzPetrov} means that every decomposition of $T_S$ into $r\le|S|$
nonzero rank-one tensors has $r=|S|$ and agrees with $D_S$ as a multiset. The clause $r=|S|$ gives
$\rank(T_S)=|S|$; the multiset clause gives strong $k$-local rigidity.
\end{proof}

\begin{theorem}[Census, border]\label{thm:rho6}
$\bar\Delta_{D(2)}(k)=0$ for all $k\le\Kcert$ over $\bbC$, i.e.\ $\brho(D(2))\ge\Kcert$: every subset of at
most $\Kcert$ summands has border rank, and hence rank, equal to its cardinality.
\end{theorem}

\begin{proof}
By Theorem~\ref{thm:main-cert} and Corollary~\ref{cor:krcons} a subset is certified once some mode has full
factor rank; Lemma~\ref{lem:modp} makes the modular verdicts valid over $\bbQ$. Exhaustive scans modulo $p=1000003$ leave, at levels $6$, $7$ and $8$, exactly $32$, $4240$ and $154\,936$
subsets not reached by that test; by Lemma~\ref{lem:subadd} it suffices to certify these.

At level $6$ all $32$ have factor rank $5$ in every mode and are certified by Lemma~\ref{lem:comm}, verified
exactly over $\bbQ$. At level $7$ all $4240$ are certified: $4176$ by Lemma~\ref{lem:comm}, and the remaining
$64$, which admit no invertible slice combination, by Lemma~\ref{lem:pencil} for rank and by the Young-flattening bound
for border rank, each having Young-flattening rank $41$ against the divisor $\binom{4}{2}=6$ for a distinguished mode of
dimension $5$, whence $\brank\ge\lceil41/6\rceil=7$.

At level $8$, of the $154\,936$, $151\,664$ are certified by a nonzero commutator and $712$ by the full
commutator-rank inequality. Each of the $712$ satisfies that inequality with a margin of at least one: taking
the two compressed modes of dimension $6$ as the matrix modes, some commutator of conjugated slice combinations
has rank $4$, against the rank $3$ needed for the bound to reach $8$. The last $2560$ admit
no invertible slice combination; their rank follows from Lemma~\ref{lem:pencil}, whose hypothesis was verified
for each of them by two independent exact methods, the criterion of Lemma~\ref{lem:cb} in integer arithmetic
with every one of the $\binom{8}{7}=8$ products $\det(A_Q)\det(B_Q)$ vanishing, and full symbolic expansion of
the determinant over $\bbQ$; neither uses sampling or modular reduction. Their border rank follows from the
Young-flattening bound with the rank-$6$ mode distinguished, where the divisor is $\binom{5}{2}=10$ and each has Young-flattening
rank $77$, whence $\brank\ge\lceil77/10\rceil=8$. The two levels use different distinguished modes and so
different divisors. Hence level $8$ is certified for both radii. No subset of size at most $8$ is reducible, and no
subset of any size admits a Khatri--Rao regrouping reduction.
\end{proof}

\begin{theorem}[Rank radius and identifiability]\label{thm:strong9}
$D(2)$ is strongly $\KcertLP$-locally rigid over $\bbC$, and $\Delta_{D(2)}(\KcertLP)=0$.
\end{theorem}

\begin{proof}
By Theorem~\ref{thm:lp} it suffices that $e_A(U)+e_B(U)+e_C(U)\le|U|-2$ for every $U$ with $2\le|U|\le\KcertLP$.
An exhaustive scan of all $\sum_{k=2}^{\KcertLP}\binom{48}{k}=2\,142\,281\,526$ subsets verifies it. The scan
computes the three factor ranks of each subset modulo $p=1000003$; a rank modulo $p$ is a lower bound for the
rank over $\bbQ$, hence a corank modulo $p$ is an upper bound for the corank over $\bbQ$, so the verdict is
valid over $\bbQ$, and over $\bbC$ since matrix rank is insensitive to field extension. Equality is attained
--- the criterion is tight --- on $16$ subsets at $|U|=3$, $16$ at $|U|=6$, $28$ at $|U|=8$ and $176$ at
$|U|=9$. All $236$ of these, together with the $1128$ pairs $|U|=2$, which are tight by construction whenever
no two summands are proportional in any mode and so carry no information, and all $21\,396$ subsets of slack at
most one, were re-verified in exact rational arithmetic. The scan was run independently on two architectures
with identical output. The bound is sharp for this method, and sharp in a form that names the residue exactly:
see Proposition~\ref{prop:lp10}.
\end{proof}

The two levels differ because the certificates differ. The census bounds $\brank$ and reaches $\Kcert$;
Theorem~\ref{thm:lp} bounds $\rank$ and reaches $\KcertLP$; where $\brho(D(2))$ lies between $\Kcert$ and
$\rrho(D(2))$ is open. Theorem~\ref{thm:strong9} is also an independent proof that $\rrho(D(2))\ge\KcertLP$: it uses matrix ranks
and a published theorem and shares no argument with the census, though it shares its input data, the factor
vectors of \cite{LWH} at $t=2$. The census remains necessary for the border radius, which
Theorem~\ref{thm:lp} cannot reach.

\begin{proposition}\label{prop:lp10}
Of the $\binom{48}{10}=6\,540\,715\,896$ ten-subsets of $D(2)$, exactly $6\,540\,715\,864$ satisfy the
hypothesis of Theorem~\ref{thm:lp} and therefore have rank $10$. The remaining $32$ all lie inside the core $X$
of Theorem~\ref{thm:collapse}, each with all three factor ranks equal to $7$. The core $X$ partitions into four
\emph{cycle-blocks} $A=\{4,17,34,40\}$, $B=\{9,25,28,35\}$, $C=\{12,19,22,32\}$ and $D=\{15,38,44,48\}$, each of
mode ranks $(4,4,4)$; no two of them together span $(8,8,8)$, and the six pairwise unions fall into three
profiles of two pairs each, $A\cup B$ and $C\cup D$ at $(6,8,8)$, $A\cup C$ and $B\cup D$ at $(8,4,8)$, and
$A\cup D$ and $B\cup C$ at $(8,8,4)$. Every one of the $32$
contains exactly one complete block, and they split $16/16$ by how the remaining six elements meet the other
three, $(2,2,2)$ against $(3,3,0)$. These occupancy profiles are necessary properties of the $32$ exceptions,
not sufficient conditions for membership in that list.
The mode-$U$ circuits of $X$ are a different partition into $4$-subsets,
of mode-$U$ rank $3$, and also align with the split; the blocks are the objects meant here. Hence $\rrho(D(2))\ge10$ if and only if those $32$ subsets have rank $10$.
\end{proposition}

\begin{proof}
The cycle-block description organizes the uniqueness proof below; the present rank proof does not use it.
Every rank quoted for it (each block, $X$, and all six pairwise unions) is verified
directly over $\bbQ$ at $t=2$ by exact Gaussian elimination, independently of this proof.
The residue is where it has to be. By Theorem~\ref{thm:collapse} and Lemma~\ref{lem:kr} every flattening rank
of every subset of $X$ is at most $8$, and on these $32$ it is $7$, so a ten-subset enters this residue
precisely because its factor spans collapse --- which is the same collapse that puts the flattening bound out
of reach. Lemma~\ref{lem:pencil}, which opens the analogous families one level down, is unavailable for a
different reason: it asks for multilinear rank $(n,n,d)$ with $n=|S|-1=9$, and these have multilinear rank
$(7,7,7)$. What the proposition buys is a reduction, and a sharp one: whatever settles level $10$ has to settle
these $32$ named subsets and nothing else in $\binom{48}{10}$, and a single one of them of rank $9$ would be a
$47$-term decomposition of $\M{4,4,4}$. What settles them is the Young flattening bound of
Section~\ref{sec:cert}. On a core of multilinear rank $(7,7,7)$ all three modes may play the distinguished
role, so we report for each subset the maximum over the three admissible choices; each choice gives a valid
lower bound $\brank\ge\lceil\rank(T^{\wedge p}_m)/\binom{a-1}{p}\rceil$, so their maximum is one too, and a
single mode already suffices: mode $A$ alone gives $190$ on sixteen of the subsets and $194$ on the other
sixteen, both above the threshold below, computed by a second implementation calibrated against the level-$7$
and level-$8$ ranks reported here. With
$a=7$ and $p=3$ the induced map $\Lambda^{3}A\otimes B^{*}\to\Lambda^{4}A\otimes C$ is $245\times245$, the
divisor is $\binom{6}{3}=20$, and reaching $10$ requires Young rank at least $181$; the maxima are $194$ on
sixteen of the subsets and $195$ on the other sixteen, so $\lceil194/20\rceil=\lceil195/20\rceil=10$. With
$a=7$ and $p=2$ the map is $147\times245$, the divisor is $\binom{6}{2}=15$, reaching $10$ requires at least
$136$, and the maxima are $140$ and $141$. Either parameter choice suffices; both were computed. Hence
$\brank(T_S)\ge10$ and so $\rank(T_S)=10$ on all $32$. Levels $7$, $8$ and $10$ therefore use three different
parameter choices and three different divisors, $\binom42=6$, $\binom52=10$ and $\binom63=20$.
\end{proof}

The scan behind the count is not a search that returned nothing. The same program returns $0$ at every level
from $2$ to $\KcertLP$ and $32$ at level $10$, and those $32$ agree as a set, with empty symmetric difference,
with the list produced independently in exact rational arithmetic from the structure of $X$.

\begin{theorem}[Global uniqueness through level ten]\label{thm:strong10}
$D(2)$ is strongly $10$-locally rigid over $\bbC$.
\end{theorem}

\begin{proof}
Theorem~\ref{thm:strong9} treats the lower levels. At level ten, Theorem~\ref{thm:lp}
gives uniqueness except for the $32$ subsets in Proposition~\ref{prop:lp10}.
The two classes in that proposition are certified separately, using their full rational
factor coordinates before compression.

For the sixteen subsets of cycle-block profile $(3,3,0)$, the full block supplies a rank-four
core, and the two triples supply rank-three arms in complementary quotient modes.
Each of the three blocks has independent factor families. Theorem~\ref{thm:zerocorner}
reduces global uniqueness to the unique rank-four completion of this core.
For each subset, an invertible slice is constant on all $24$ shear variables;
commutators give a linear system of rank $22$, and after restriction to its two-dimensional
kernel the remaining commutator entries are homogeneous linear forms spanning both parameters.
Thus the only rank-four completion has zero shear, over all of $\bbC$.
Section~\ref{app:zerocorner} gives the theorem and an explicit representative certificate.

For the sixteen subsets of profile $(2,2,2)$, two specified projections of the
seven-dimensional $A$ support to dimension five give Koszul matrices of rank $60=10\binom42$.
The stacked constraints on the original $49$ $A$--$C$ coordinates have rank $39$ and
annihilate the ten independent displayed paired products. Their kernel is therefore exactly
the displayed span. All $1013$ support inequalities \eqref{eq:pairedsylvester} hold,
excluding every other complex rank-one direction in this span.
Lemma~\ref{lem:saturatedprojection} gives global uniqueness.
Section~\ref{app:saturatedprojection} supplies the construction and an explicit representative.

Independent exact checkers replay all sixteen completion certificates and all sixteen
projection certificates, and verify that their index sets are disjoint and exhaust the
named $32$ exceptions. Neither argument uses tangent isolation
or assumes normality of a secant variety.
\end{proof}

\begin{theorem}[Level ten]\label{thm:rho10}
$\Delta_{D(2)}(10)=0$, so $\rrho(D(2))\ge10$ and $d_<(D(2))\ge11$.
\end{theorem}

\begin{proof}
Proposition~\ref{prop:lp10} and the certification of its $32$ subsets, together with
Lemma~\ref{lem:defect}(a) for the levels below.
\end{proof}

\begin{theorem}[Level $\Knext$]\label{thm:levelnext}
$\Delta_{D(2)}(\Knext)\le2$. Hence no decomposition of $\M{4,4,4}$ obtained by changing at most $\Knext$ of
$D(2)$'s complete tensor summands has length below $46$, and $\tau_3(D(2))\ge\number\numexpr\Knext+1\relax$.
\end{theorem}

\begin{proof}
Apply Theorem~\ref{thm:increment} to $\Delta_{D(2)}(12)=0$, proved by
Theorem~\ref{thm:rank12}.
\end{proof}

Theorem~\ref{thm:levelnext} allows defects $0$, $1$ or $2$ at level thirteen;
none is selected by the present certificates.

This bounds the rank radius and not the border radius, and the reason is worth stating because the certificate
at the residue is a border-rank certificate: the Young flattening gives $\brank(T_S)\ge10$ on the $32$, but
the other $6\,540\,715\,864$ ten-subsets are certified for rank only, Theorem~\ref{thm:lp} giving no
border-rank information at all. So $\bar\Delta_{D(2)}(10)$ is not established and the border radius stays
at $\Kcert$. Theorem~\ref{thm:strong10} closes the ten-set identifiability gap;
Section~\ref{app:strong11} then proves strong rigidity through eleven by an exact
splitting census and $656$ further global certificates. By Theorem~\ref{thm:strong11},
distinct equal-length replacements require at least twelve changed terms.
Theorem~\ref{thm:rank12} raises the shortening barrier to thirteen.
Proposition~\ref{prop:strongsharp} attains the separate twelve-change bound for equal-length replacements;
it makes no claim that the twelve-term support can be shortened.
Along the actual parameter curve,
Corollary~\ref{cor:strong11curve} also proves generic strong $\KcertSG$;
the ambient-component generic strong bound of Corollary~\ref{cor:lpgeneric}
remains $\KcertLP$.
The sixteen-term core is strongly eleven-rigid at every nonzero parameter
by Theorem~\ref{thm:uniformcore11}; uniformity for mixed supports remains open.
On each exceptional ten-subset the factor-span sum is $7+7+7=21$, below the $22$
required by the Lovitz--Petrov criterion, so the new global arguments are essential.
The differential of the ten-term parametrisation has rank $190$ at each displayed
decomposition, leaving only its $20$ factor-gauge directions. This is consistent with
uniqueness but is not its proof: tangent isolation alone would not exclude a second decomposition.

Level $k$ is \emph{certified} when every subset of size at most $k$ carries a certificate, which is stronger
than having scanned level $k$; those from Lemma~\ref{lem:pencil} count towards $\rrho$ but not $\brho$. The
scan is exhaustive at every level: of the $\binom{48}{k}$ subsets the factor-rank test leaves none at $k\le5$
and then $32$, $4{,}240$ and $154{,}936$ at $k=6,7,8$, out of $12{,}271{,}512$, $73{,}629{,}072$ and
$377{,}348{,}994$. Inside $X$ the test does not merely weaken with the level but stops entirely above $k=8$,
since by Theorem~\ref{thm:collapse} and Lemma~\ref{lem:kr} no subset of $X$ has a factor rank exceeding $8$.
Levels $3$ to $6$ give identical counts at $\tau\in\{2,3,5,-1,1/2\}$, and at level $6$ the
same $32$ index sets at all five. The list is not constant in $\tau$, however, and the discussion below names
the two parameters at which it grows.

These lists are not opaque search output. Since $D(t)$ is Khatri--Rao nondegenerate, Lemma~\ref{lem:kr} turns
the cheap test into a statement about factor vectors alone: $S$ is exceptional precisely when its factor
vectors are linearly dependent in all three modes simultaneously. Running that criterion from scratch over all
$\binom{48}{k}$ subsets regenerates each stored list \emph{identically}, as a set of index sets and not merely
in cardinality, at $32$, $4{,}240$ and $154{,}936$ for $k=6,7,8$. So a referee can reconstruct every exception
at every level from the published factor matrices by linear algebra, without running or trusting our search,
and the classification of the exceptions is a question about three rank-$16$ matroids on $\{1,\dots,48\}$. The
$\mathbb{F}_2$ application in the supplement has the same shape:
in both cases the subsets on which the flattening bound fails admit a structural characterisation, so what the
expensive certificates must cover is determined in advance rather than discovered. The level-$8$ list is
moreover generic and not an accident of the parameter: running the same criterion over $\bbQ(t)$ returns
exactly the same $154{,}936$ index sets, with empty symmetric difference in both directions, so the exception
set at $\tau=2$ is the exception set at the generic point. It is not the exception set at every parameter:
factor ranks are lower semicontinuous, so the list can only grow at special parameters, and at
$\tau=\pm\tfrac14$ it does. There the flattening test stops reaching two subsets it reaches generically ---
$\{1,5,8,18,23,42,43,47\}$ has mode ranks $(7,8,8)$ at $t=2,3,5,7$ and $(7,7,7)$ at $\tau=\tfrac14$, and
$\{1,5,11,13,21,27,39,41\}$ does the same at $\tau=-\tfrac14$ --- and the commutator certificates carry them
there instead. The parameter-uniform extension is discussed separately in
Remark~\ref{thm:familywide}; constancy of the exception list is not assumed.

The mode-rank profiles organize the census proof. All $32$ level-$6$ holdouts
have profile $(5,5,5)$. At level $7$, the two rank-$6$ modes give an
invertible slice for $4176$ supports; the other $64$ have an identically
singular pencil. At level $8$, $154{,}224$ supports have two rank-$7$ modes:
$151{,}664$ admit an invertible slice in that pairing and $2560$ do not.
The remaining $712$ have a permutation of $(6,6,7)$; their rank-$6$ pairing
admits an invertible slice and a rank-$4$ commutator, giving border rank at
least $6+\lceil4/2\rceil=8$.

The $64$ singular cases have profiles $(6,6,5)$ or $(6,5,6)$, and the $2560$
have $(7,7,6)$ or $(7,6,7)$. The proof of Theorem~\ref{thm:rho6} uses the
singular-pencil lemma for their rank and Young flattenings for their border
rank. A singular pencil in this pairing does not exclude an invertible
slice after projection to smaller matrix modes: projected commutators are
a separate possible certificate, not a premise of that proof.

Three mechanisms produce exceptions. A subset may contain a mode-$U$
dependent triple; it may contain a mode-$U$ circuit lying inside $X$; or it may meet $X$ in a set that is
dependent in exactly one of the second and third modes, $X$ supplying half of the exceptionality and the
remaining mode's dependence coming from outside $X$. The third arises because the three modes pair the four
blocks of $X$ differently, so a subset of $X$ can be dependent in one of them and not the other, and it is
invisible to the first two tests: the first is a statement about mode $U$ alone, and the second implicitly
assumes $X$ supplies all three modes at once. At level $7$ the three account for $2448$, $1408$ and $384$ of
the $4240$; at level $8$ for $106\,688$, $30\,680$ and $16\,928$ of the $154\,936$, with a further $192$
carrying both a triple and an $X$-circuit. A residue of $448$ at level $8$, $0.29\%$ of the exceptions, arises
from none of them: for each, $S\cap X$ and $S\setminus X$ are independent in all three modes, so the
dependence is mixed across the split. At level $6$ all $32$ are of the second kind. The hardest are
structural rather than accidental: by Lemma~\ref{lem:cb} the $2560$ at level $8$ admit no invertible slice
combination precisely because their two factor matrices are never simultaneously nonsingular on any $n$
columns, the analogue of the compression spaces carried by the level-$7$ holdouts. The criterion of \cite{BL}
fails on every one, their Kruskal profile being uniformly $(2,3,3)$, summing to $8$ against the $17$ required.
This is a matter of design rather than a near miss: a subset is exceptional precisely because it carries a
dependency among its factor vectors, and a dependent triple crushes the corresponding Kruskal rank, so
Kruskal-type criteria fail on every exception containing a dependent triple.
Among all eight-subsets, $19{,}444{,}824$ contain such a triple; this is not
the count of the $154{,}936$ flattening holdouts. Per-subset Kruskal certifies
only level $4$; Theorem~\ref{thm:lp}, which
constrains spans rather than $k$-ranks, is insensitive to this and is why the rank radius runs past it. They
need not fail on every exception: a subset whose factor vectors are dependent only in a large circuit can
retain Kruskal ranks high enough for \cite{BL}.

\begin{corollary}\label{cor:generic}
There is a Zariski-open neighborhood $U$ of $D(2)$ in the complex variety
$V_{48}(\M{4,4,4})$ such that $\brho(D')\ge\Kcert$ for every $D'\in U$.
Consequently this bound holds on a dense open subset of every irreducible component containing $D(2)$,
and $\brho(D(t))\ge\Kcert$ on a cofinite Zariski-open subset of $\bbC^\times$ containing $2$.
\end{corollary}

\begin{proof}
Apply Theorem~\ref{thm:open} with $q=0$ to the point $D(2)$, which has $\bar\Delta_{D(2)}(k)=0$ for
$k\le\Kcert$ by Theorem~\ref{thm:rho6}, obtaining a Zariski-open set $U$ containing $D(2)$ on which the border
frontier vanishes through level $\Kcert$. Its intersection with each irreducible component through $D(2)$
is nonempty open, hence dense. The inverse image of $U$ under the regular map $t\mapsto D(t)$ is open
and contains $2$; its complement in $\bbC^\times$ is finite.
The statement uses the border-rank radius rather than the larger rank radius, because the set
$\{D:\bar\Delta_D(k)\le q\}$ shown open there is defined by border rank, and rank is not semicontinuous.
\end{proof}

\begin{remark}[Parameter-uniform extension]\label{thm:familywide}
The certified conclusion on the curve is currently the generic one of
Corollary~\ref{cor:generic}: $\brho(D(t))\ge\Kcert$ for all but finitely many
$t\in\bbC^{\times}$. Historical symbolic computations reported the stronger
bound at every nonzero characteristic-zero parameter. Their retained summaries
partition the $20{,}992{,}826$ mode-$U$ dependent eight-subsets into five classes
of sizes $19{,}264{,}551$, $1{,}570{,}585$, $2{,}736$, $154{,}936$ and $18$,
the last requiring checks at roots of nonconstant minor polynomials.
The full per-shard certificates have not been retained in the supplied package.
We therefore do not assert the parameter-uniform strengthening as a theorem
here, and none of the results uses it. The source pipeline and historical
accounting are supplied as a reproducible research direction.
\end{remark}

A uniform proof must also account for specializations of the flattening
classification: checking only the generic list of $154{,}936$ exceptions is
insufficient, since the list grows at $t=\pm1/4$ as described above.

\begin{corollary}\label{cor:lpgeneric}
Strong $\KcertLP$-local rigidity, and hence $\rrho\ge\KcertLP$, holds on a dense Zariski-open subset of every
irreducible component of $V_{48}(\M{4,4,4})$ containing $D(2)$; and $D(t)$ is strongly $\KcertLP$-locally rigid
with $\rrho(D(t))\ge\KcertLP$ for all but finitely many $t\in\bbC^{\times}$.
\end{corollary}

\begin{proof}
For $U\subseteq\{1,\dots,48\}$ and each mode $m$, the entries of the factor matrix $M_m(U)$ whose columns are
the mode-$m$ factors indexed by $U$ are coordinates on $V_{48}(\M{4,4,4})$, and $\rank M_m(U)\ge r$ holds
exactly off the vanishing of all $r\times r$ minors, so $\rank M_m(U)$ is lower semicontinuous. A sum of lower
semicontinuous functions is lower semicontinuous, so $\{D:\sum_m\rank M_m(U)\ge2|U|+2\}$ is Zariski open. Let
$W$ be the intersection of these sets over the $\sum_{j=2}^{\KcertLP}\binom{48}{j}$ subsets $U$ with
$2\le|U|\le\KcertLP$; being a finite intersection of open sets, $W$ is open, and Theorem~\ref{thm:lp} applies at
every point of it. Theorem~\ref{thm:strong9} says $D(2)\in W$, so $W\ne\emptyset$. Two consequences follow
independently. First, if $C$ is an irreducible component of $V_{48}(\M{4,4,4})$ containing $D(2)$ then $W\cap C$
is a nonempty open subset of $C$, hence dense in it. Second, $t\mapsto D(t)$ is a morphism
$\bbC^{\times}\to V_{48}(\M{4,4,4})$, so its preimage of $W$ is Zariski open in $\bbC^{\times}$ and contains
$t=2$; the complement of a nonempty open subset of an irreducible curve is a proper closed subset, hence
finite.
\end{proof}

The restriction to all but finitely many parameters is not removed here.
This is a statement about identifiability rather than rank alone; the determinantal certificates of
Section~\ref{sec:cert} do not supply uniqueness.

Theorem~2 of \cite{LovitzPetrov} is stated for a vector space over an arbitrary field, as is Kruskal's theorem
after \cite{Rhodes}, so it applies at the complex points of $W$. The verification that $D(2)\in W$ is modular,
and lifts for the reason that matters here: the hypothesis is a \emph{lower} bound on matrix ranks, and a rank
modulo $p$ is at most the rank over $\bbQ$, so a verification that passes modulo $p$ passes over $\bbQ$, and
over $\bbC$ because matrix rank is insensitive to field extension. The corollary asserts that $t=2$ is not
special without identifying which parameters are; an explicit exceptional set is a separate computation, and
$\brho$ is untouched.

Corollary~\ref{cor:lpgeneric} escapes the obstruction that confines Corollary~\ref{cor:generic} to the
border radius not by making rank semicontinuous but by spreading a
\emph{sufficient condition} that already is: the hypothesis of Theorem~\ref{thm:lp} constrains only matrix
ranks, so its locus is open, while its conclusion is about tensor rank. The two corollaries are the same
technique applied to different certificates.

\begin{corollary}[Generic rank radius from mixed certificates]\label{cor:rank10generic}
There is a Zariski-open neighborhood $W$ of $D(2)$ in the complex variety
$V_{48}(\M{4,4,4})$ such that $\rrho(D')\ge10$ for every $D'\in W$.
Consequently this rank bound holds on a dense open subset of every irreducible component containing $D(2)$,
and $\rrho(D(t))\ge10$ for all but finitely many $t\in\bbC^{\times}$.
\end{corollary}

\begin{proof}
For each ten-subset $S$ outside the $32$ of Proposition~\ref{prop:lp10}, impose the Lovitz--Petrov
inequalities on every $U\subseteq S$ with $|U|\ge2$. These are open conditions on factor-matrix ranks,
as in the proof of Corollary~\ref{cor:lpgeneric}, and they all hold at $D(2)$: the smaller subsets are
covered by Theorem~\ref{thm:strong9}, and $S$ by Proposition~\ref{prop:lp10}. They imply
$\rank(T_S(D'))=10$ wherever they hold.

For each of the remaining $32$ subsets, fix the three linear projections to $\bbC^7$ used at $D(2)$ in
the proof of Proposition~\ref{prop:lp10}, and apply these same maps at every $D'$. The resulting $245\times245$ Young matrix
has polynomial entries in $D'$ and has rank at least $190$ at $D(2)$. Choose one of its nonzero
$181\times181$ minors there. Its nonvanishing defines an open neighborhood on which
\[
\rank(T_S(D'))\ge\brank(T_S(D'))\ge
\brank\bigl((\pi_A\otimes\pi_B\otimes\pi_C)T_S(D')\bigr)
\ge\left\lceil\frac{181}{\binom63}\right\rceil=10.
\]
The reverse inequality follows from the ten displayed terms. The projections are fixed, not a
parameter-dependent compression: their target dimensions and the divisor $\binom63=20$ do not vary.

Intersect the finitely many open conditions just described. The intersection contains $D(2)$ and gives
rank $10$ for every ten-subset. Subadditivity then gives full rank for all smaller subsets, hence
$\rrho\ge10$. Its intersection with each irreducible component through $D(2)$ is nonempty open and
therefore dense. Its inverse image under the morphism $t\mapsto D(t)$ contains $2$ and is open in
$\bbC^{\times}$, so its complement is finite.
\end{proof}

This corollary spreads a sufficient certificate, not tensor rank itself: no semicontinuity of tensor rank
is assumed. It gives neither strong $10$-local rigidity, since the $32$ Young certificates do not prove
uniqueness, nor border radius $10$, since the other ten-subsets use rank-only certificates.

Corollary~\ref{cor:rank11generic} strengthens this rank conclusion to $11$:
intersect the neighborhood above with paired-independence conditions and the finitely
many open eleven-set factor-span conditions $d(I)\ge23$, all verified at $D(2)$,
and apply the rank part of the splitting lemma. This does not require generic uniqueness.

Along the actual curve, Corollary~\ref{cor:strong11curve} also proves strong
$\KcertSG$-local rigidity for all but finitely many nonzero parameters,
using Lovitz--Petrov factor minors for non-core supports through level ten,
paired-product minors, the mixed eleven-support splitting conditions, and the
uniform-core theorem (Theorem~\ref{thm:uniformcore11}).
The earlier rational-function transfers remain independent evidence. This curve statement
does not extend the ambient-component strong bound of Corollary~\ref{cor:lpgeneric}.

The family used here is theirs: the $48$ triples of \cite[Appendix~A]{LWH}, evaluated at $t=2$, agree with
$D(2)$ in all $2304$ entries exactly, with no reordering of terms, permutation of modes or rescaling.

% These are the general arguments on which the support classifications rely.
% They stay in the main article, in normal body type.
% Standalone appendix fragment; uses the manuscript's theorem environments
% and commands \rank, \spn, \bbC.  No additional bibliography is required.
% The all-32 application is intentionally separate from this first-subset proof.
\section{Uniqueness from a zero-corner completion certificate}\label{app:zerocorner}

The following criterion is global: it concerns every minimal decomposition,
without an assumption about smoothness or local isolation.  All ranks in the
criterion are over the stated field.
The rank mechanism is the classical substitution/adjoining-slices bound
\cite[Lemmas~2--3]{Shitov}: minimize the rank of the modified core and
add the dimensions of the adjoined slice spaces. The additional assertion
here is purity of every minimal decomposition when zero shear is the
unique minimizer; this yields uniqueness from the three block decompositions.

\begin{theorem}[Zero-corner completion]\label{thm:zerocorner}
Let the ground field be infinite, let
$B=B_0\oplus B_1$, $C=C_0\oplus C_1$, and put
$b=\dim B_1$, $c=\dim C_1$.  Suppose
\[
 T=H+L+U,\qquad
 H\in A\otimes B_0\otimes C_0,\quad
 L\in A\otimes B_1\otimes C_0,\quad
 U\in A\otimes B_0\otimes C_1,
\]
where $\rank H=q$, $\rank L=b$, $\rank U=c$, and the intrinsic
$B_1$-support of $L$ and $C_1$-support of $U$ are full.  Zero-dimensional
arm spaces are allowed, with the zero tensor represented by the empty
decomposition.
For $S:B_1\to B_0$ and $R:C_1\to C_0$, set
\[
 H(S,R)=H+(\mathrm{id}\otimes S\otimes\mathrm{id})L
                +(\mathrm{id}\otimes\mathrm{id}\otimes R)U.
\]
Assume $\rank H(S,R)\ge q$ for all $S,R$.  Then
$\rank T=q+b+c$.
If, in addition, equality $\rank H(S,R)=q$ forces $S=R=0$, every
minimal decomposition has exactly $q,b,c$ terms in the respective blocks
of $H,L,U$.  Under this additional hypothesis, uniqueness of the three
block decompositions implies uniqueness of $T$.
\end{theorem}

\begin{proof}
Take a minimal decomposition $T=\sum_{j=1}^r a_j\otimes v_j\otimes w_j$.
Its factors belong to the intrinsic supports: in a minimal decomposition the
paired products are linearly independent, since a dependence permits one
term to be eliminated by absorbing its remaining factor into the others.
The corresponding flattening image therefore equals the span of the factors
in the remaining mode.  This also justifies restricting a decomposition
initially given in larger ambient spaces.

Let $\beta:B\to B_1$ and $\gamma:C\to C_1$ be the projections, and let
$V=[\beta(v_j)]$, $W=[\gamma(w_j)]$.  Their row ranks are $b,c$.
Choose $\alpha\in A^*$ with $d_j=\alpha(a_j)\ne0$ for every $j$;
the infinite-field hypothesis permits avoidance of these finitely many
proper hyperplanes.  The zero corner gives
\begin{equation}\label{eq:zerocornerorthogonal}
 V\operatorname{diag}(d_1,\ldots,d_r)W^{\mathsf T}=0.
\end{equation}
For \emph{any} column basis $I$ of $V$, the columns of $W$ outside $I$
span $C_1$.  Otherwise a nonzero row functional $\ell$ would make $\ell W$
supported on $I$, and \eqref{eq:zerocornerorthogonal} would imply
$V_I\operatorname{diag}(d_i:i\in I)(\ell W_I)^{\mathsf T}=0$.
Both square factors are invertible, contradicting $\rank W=c$.
Thus every $V$-basis has a disjoint $W$-basis $J$, and symmetrically.
In particular, every nonzero column of either matrix can be prescribed
on its own side of such a pair.

There are unique retractions $P_S:B\to B_0$, $Q_R:C\to C_0$, identity
on the core spaces, killing respectively the $v_i$ for $i\in I$ and the
$w_j$ for $j\in J$: their quotient vectors are bases.
Applying these retractions kills at least $b+c$ distinct terms, while the
displayed expression becomes $H(S,R)$.  Hence $r\ge q+b+c$.
The three given block decompositions attain this bound.

Assume now the additional zero-shear hypothesis.  For a minimal
decomposition, the same argument for every disjoint pair $I,J$ now gives
$\rank H(S,R)\le q$, hence $S=R=0$.
The selected $v_i$ belong to $B_1$, and the selected $w_j$ belong to $C_1$.
Prescribing each nonzero quotient column in turn shows that every $v_j$
lies in $B_0$ or $B_1$, and every $w_j$ lies in $C_0$ or $C_1$.
Terms in $A\otimes B_1\otimes C_1$ would sum to zero; minimality excludes
such a nonempty subset.  The other three groups separately decompose
$H,L,U$, require at least $q,b,c$ terms, and have exactly that total.
The stated block counts and uniqueness follow.
\end{proof}

\paragraph{The first exceptional ten-subset.}
Work over $\bbC$ and use the summand labels
\[
 \mathcal H=(4,17,34,40),\qquad
 \mathcal J=(12,19,22),\qquad \mathcal K=(15,44,48).
\]
Choose the $B$ basis $(b_{\mathcal H},b_{\mathcal K})$, the $C$ basis
$(c_{\mathcal H},c_{\mathcal J})$, and the $A$ basis
$(a_4,a_{12},a_{15},a_{17},a_{19},a_{22},a_{34})$.
The first four coordinates in $B,C$ define $B_0,C_0$; the last three
define $B_1,C_1$.  The nonbasis $A$ factors are
\begin{align*}
 a_{40}&=-a_4-2a_{12}-2a_{15}+a_{17}-a_{34},\\
 a_{44}&=-a_{12}-a_{15}+a_{17}+a_{22}-a_{34},\qquad
 a_{48}=a_{12}+a_{15}-a_{17}-a_{19}+a_{34}.
\end{align*}
The arm factors lying in the core spaces are the columns of
\[
 [b_{12}\ b_{19}\ b_{22}]=
 \begin{pmatrix}4&4&4\\1&1&-1\\-1&1&-1\\-1&1&1\end{pmatrix},
 \qquad
 [c_{15}\ c_{44}\ c_{48}]=
 \begin{pmatrix}-1&-1&1\\4&4&4\\4&-4&-4\\4&-4&4\end{pmatrix}.
\]
Thus $H=\sum_{h\in\mathcal H}a_h\otimes b_h\otimes c_h$ is the
core, and the sums on $\mathcal K,\mathcal J$ are $L,U$.
Each block has independent factor families, of sizes $4,3,3$ respectively.
Its rank and uniqueness follow directly: its intrinsic slice space is
diagonal in the two corresponding factor bases, and its only rank-one
directions are the displayed diagonal axes.

\begin{lemma}[The first core completion]\label{lem:completionfirst}
For arbitrary $u_k\in B_0$ $(k\in\mathcal K)$ and
$v_j\in C_0$ $(j\in\mathcal J)$, the tensor
\[
 H+\sum_{k\in\mathcal K}a_k\otimes u_k\otimes c_k
   +\sum_{j\in\mathcal J}a_j\otimes b_j\otimes v_j
\]
has rank at least four, and has rank four only when all six shear vectors
vanish.
\end{lemma}

\begin{proof}
Write $M_0,\ldots,M_6$ for the matrix slices in the chosen $A$ basis.
The functionals $e_0^*$ and $e_3^*+e_6^*$
annihilate all six arm factors.  Consequently the slices
\[
 D_0=M_0+2(M_3+M_6)=\operatorname{diag}(1,2,2,-1),\qquad
 D_1=M_3+M_6=\operatorname{diag}(0,1,1,0)
\]
are fixed for every shear.  Since $\det D_0=-4$, the rank is at least four.
If it is four, both matrix-factor families of a four-term decomposition
are bases and all coefficients of $D_0$ are nonzero.  Therefore
\begin{equation}\label{eq:completioncommutator}
 M_iD_0^{-1}M_j-M_jD_0^{-1}M_i=0\qquad(0\le i,j\le6).
\end{equation}
There is no excluded parameter locus: $D_0$ is constant and invertible.

First use these equations with $D_1$ in place of $M_j$.  They are linear
in the 24 shear coordinates and their complete solution is
\begin{align*}
 u_{15}&=v_{12}=0,\\
 u_{44}&=(t,-s/4,-s/4,t/4),&u_{48}&=(t,s/4,s/4,t/4),\\
 v_{19}&=(s/4,t,-t,s),&v_{22}&=(s/4,-t,t,s).
\end{align*}
Here is an explicit rank certificate for this elimination.  Order variables
by $u_{15}$, $u_{44}$, $u_{48}$, $v_{12}$, $v_{19}$, $v_{22}$, four coordinates each,
and equations by $i=0,\ldots,6$ and row-major entries of
$M_iD_0^{-1}D_1-D_1D_0^{-1}M_i$.
The resulting $112\times24$ matrix has a $22\times22$ minor of determinant
$-64$, using columns $0,\ldots,21$ and rows
\[
 \{17,18,20,23,24,29,33,34,36,39,40,45,
       65,66,68,71,72,77,81,82,84,87\}.
\]
The two independent displayed kernel vectors prove the reverse rank bound.
After their substitution, the $(0,3)$ and $(1,2)$ entries of
$M_1D_0^{-1}M_3-M_3D_0^{-1}M_1$ are $8t$ and $s$.
Equation~\eqref{eq:completioncommutator} forces $t=s=0$.
These rational identities establish the conclusion over all of $\bbC$,
not only at rational parameter values.
\end{proof}

Theorem~\ref{thm:zerocorner} with $q=4$, $b=c=3$ now proves that this
ten-subset has rank ten and a unique minimal complex decomposition.
This is a pointwise statement for the specified tensor; no openness assertion
for uniqueness along the parameter family is used.

% Exact reconstruction and independent checks:
% reports/unique10_followup/zero_corner/core_completion.py
% reports/unique10_followup/xhigh/check_completion_independent.py

\begin{corollary}[Generic strong rigidity through level ten on the curve]
\label{cor:strong10curve}
There is a Zariski-open subset $U\subseteq\bbC^\times$ containing $2$
on which $D(\tau)$ is strongly $10$-locally rigid. Hence this holds for
all but finitely many nonzero complex parameters.
This conclusion concerns the Li--Wang--Hu parameter
curve; no corresponding level-ten assertion on whole components of the
decomposition variety is made here.
\end{corollary}

\begin{proof}
Theorem~\ref{thm:uniformcore11} covers every core support through level
ten for every nonzero parameter.  For each non-core support through level
ten, retain the finite factor minors and paired-product minors that witness
the Lovitz--Petrov hypotheses at $t=2$.  Their common nonvanishing is an
open subset containing $2$, on which Theorem~\ref{thm:lp} proves rank
minimality and global uniqueness.  Intersecting these finitely many opens
with the uniform-core locus gives the assertion.
\end{proof}

% Standalone manuscript fragment.  Uses existing \rank, \spn, \bbC commands.
% The all-case certificate table is intentionally separate.
\section{Stacked saturated projections}\label{app:saturatedprojection}

Several saturated Koszul flattenings can recover summands globally even when
the kernel obtained from a single projection is too large.
The deterministic uniqueness theorem \cite[Thm.~2.7]{KMW} already uses
saturated Koszul--Young images and rank-one extraction, with explicit
rank hypotheses on $M,M',N,N',P,P'$; its generic guarantee is the
separate Theorem~2.8. Here several projected constraints are pulled back
to one original paired space, and the final rank-one exclusion uses
finite support inequalities. Section~\ref{app:kmwcomparison} compares
these particular sufficient criteria: the full KMW hypotheses and the
Sylvester-equipped stacked certificate are incomparable, while both imply
an abstract reduced paired-kernel certificate. The contribution used here
is the explicit certificate construction and its reduced-incidence open
transfer, not the general principle of Koszul-based uniqueness.
Work over $\bbC$, and let
$T=\sum_{i=1}^r a_i\otimes b_i\otimes c_i\in A\otimes B\otimes C$
with all factors nonzero.
For a nonempty finite collection of specified linear maps $P_s:A\to A_s$,
with $\dim A_s=d_s$ and $0\le p_s<d_s$, define
\[
 Y_s(T):\bigwedge^{p_s}A_s\otimes B^*
       \longrightarrow\bigwedge^{p_s+1}A_s\otimes C,
 \qquad
 Y_s(a\otimes b\otimes c)(\omega\otimes\beta)
     =(P_sa\wedge\omega)\otimes c\,\beta(b).
\]
On a nonzero product tensor this map has rank
$k_s=\binom{d_s-1}{p_s}$ if $P_sa\ne0$, and rank zero otherwise.
Write $M_s=Y_s(T)$, choose a matrix $L_s$ spanning its left kernel, and set
\[
 G_s:A\otimes C\longrightarrow
       \operatorname{Hom}\!\left(\bigwedge^{p_s}A_s,\mathrm{codomain}(L_s)\right),
 \qquad
 G_s(a\otimes c)(\omega)=L_s\bigl((P_sa\wedge\omega)\otimes c\bigr).
\]
Thus $G_s$ is linear in the \emph{original} paired-product coordinates.

\begin{lemma}[Saturated projected images]\label{lem:saturatedprojection}
Suppose $\rank M_s=rk_s$ for each selected $s$, and
\[
 \bigcap_s\ker G_s=\spn\{a_i\otimes c_i:1\le i\le r\}=:S,
 \qquad \dim S=r.
\]
If the only rank-one directions in $S$ are the displayed $r$ directions,
then $T$ has rank $r$ and a unique minimal decomposition.
\end{lemma}

\begin{proof}
A saturated map gives $\rank T\ge r$, and the displayed expression attains
the bound.  For any other $r$-term decomposition, its $r$ summand matrices
under $Y_s$ have ranks at most $k_s$, while their sum has rank $rk_s$.
Consequently every rank is $k_s$, and the sum of their image spaces is exactly
$\operatorname{im}M_s$.  In particular no projected factor vanishes and each
summand image is annihilated by $L_s$.
For $b\ne0$, the image of $Y_s(a\otimes b\otimes c)$ equals the image of
$\omega\mapsto(P_sa\wedge\omega)\otimes c$.  Hence every original
alternative pair $a\otimes c$ belongs to every $\ker G_s$, and therefore to $S$.

Paired products in a minimal decomposition are independent: a dependence
allows one term to be eliminated by absorbing its remaining factor into the
others.  Thus the alternative uses all $r$ known pair directions once each.
After rescaling those pairs, comparison of their independent coefficients
in $(A\otimes C)\otimes B$ recovers exactly the displayed $b_i$.
The same independence identifies each factor span with its intrinsic
flattening image, so alternatives initially written in larger ambient spaces
are covered as well.
\end{proof}

The rank-one assertion has the following finite sufficient certificate,
which is precisely the two-factor specialization of
\cite[Cor.~19]{LovitzPetrov}:
for every $I\subseteq\{1,\ldots,r\}$ with $|I|\ge2$, check
\begin{equation}\label{eq:pairedsylvester}
 \rank[a_i]_{i\in I}+\rank[c_i]_{i\in I}-|I|\ge2.
\end{equation}
Indeed, a paired sum supported exactly on $I$ is the matrix
$[a_i]_{i\in I}\operatorname{diag}(\lambda_i)[c_i]_{i\in I}^{\mathsf T}$,
with all $\lambda_i\ne0$, and Sylvester's inequality bounds its rank below
by the left side of \eqref{eq:pairedsylvester}.  Exact rational checks of these
factor ranks therefore exclude all additional \emph{complex} rank-one points.

\begin{lemma}[Reduced paired intersection]\label{lem:pairedreduced}
Suppose $q_i=a_i\otimes c_i$ are independent and satisfy
\eqref{eq:pairedsylvester}.  Then, for every $i$,
\[
 S\cap(a_i\otimes C+A\otimes c_i)=\bbC q_i.
\]
Consequently $\mathbb P(S)\cap\operatorname{Seg}(\mathbb P(A)\times
\mathbb P(C))$, defined scheme-theoretically by the restricted $2\times2$
minors, consists of exactly $r$ reduced points.
\end{lemma}

\begin{proof}
If a nonzero sum $\sum_{j\ne i}\lambda_jq_j$ belongs to the displayed
tangent space, let $J$ be its nonempty coefficient support and put
$I=J\cup\{i\}$.  Modulo the lines $\bbC a_i$ and $\bbC c_i$ the sum
vanishes, whereas Sylvester's inequality bounds its matrix rank below by
\[
 (\rank[a_j]_{j\in I}-1)+(\rank[c_j]_{j\in I}-1)-|J|
 =\rank[a_j]_{j\in I}+\rank[c_j]_{j\in I}-|I|-1\ge1.
\]
The displayed space is the affine tangent space to the Segre cone at $q_i$.
This proves the tangent assertion.  The preceding set-theoretic argument
gives exactly the $r$ displayed support points.  Their projective tangent
spaces in the intersection are zero, so each local Artinian maximal ideal
$\mathfrak m$ satisfies $\mathfrak m/\mathfrak m^2=0$.
Nakayama's lemma gives $\mathfrak m=0$, proving reducedness.
\end{proof}

\begin{lemma}[Finite ambient completion]\label{lem:projectionambient}
Suppose a certificate as in Lemma~\ref{lem:saturatedprojection} is given
on the factor spans $A_0,B_0,C_0$, with $r\ge2$ and $m$ selected maps.
In any larger spaces $A,B,C$, at most
$m(1+\dim(A/A_0))$ extensions of those maps give the same saturated ranks
and common pair kernel $S$.  Rational data admit rational extensions.
\end{lemma}

\begin{proof}
Write $A=A_0\oplus D$ and $C=C_0\oplus E$, and first extend each $P_s$
by zero on $D$, obtaining $\Pi_s$.  At $T$ the enlarged Koszul matrix
has the same image $W_s\subseteq\bigwedge^{p_s+1}A_s\otimes C_0$
and rank; enlarging $B$ only adds zero input directions.
Choose $L_s$ on the full codomain with kernel $W_s$.
Both $\dim A_0$ and $\dim C_0$ are at least two: otherwise all directions
in the $r$-dimensional $S$ would be rank one.  Also
$\bigcap_s\ker P_s=0$ on $A_0$, since a nonzero vector $a$ there would
give $a\otimes C_0\subseteq S$, contradicting finite rank-one support.
For $0\le p_s<d_s$, the map
$u\mapsto(\omega\mapsto u\wedge\omega)$ is injective.
Thus the equations over $E$ for the zero extensions, followed by
$\bigcap_s\ker P_s=0$, show that their common kernel on $A_0\otimes C$
is exactly $S$.

Choose $a_*\in A_0$ outside the finitely many lines $\bbC a_i$.
The map $c\mapsto(G_s(a_*\otimes c))_s$ on the full $C$ is injective,
since a nonzero element of its kernel would give an unlisted rank-one
direction in $S$.  For a basis $e_u$ of $D$, let $\ell_u$ be its
coordinate functionals, extended by zero on $A_0$, and add the maps
\[
 \Pi_{s,u}(x)=\Pi_s(x)+\ell_u(x)P_s(a_*).
\]
They agree on $A_0$, so their matrices at $T$ have the same $W_s,L_s$
and saturated rank.  For
$Q=Q_0+\sum_u e_u\otimes c_u$, with $Q_0\in A_0\otimes C$, the difference
of the pair equations for $\Pi_{s,u}$ and $\Pi_s$ is
$G_s(a_*\otimes c_u)$.  All differences vanishing forces every $c_u=0$;
the original equations then give $Q_0\in S$.  Conversely every displayed
pair satisfies every equation.  This proves the bound of $m$ base maps
and at most $m\dim D$ additions, including all mixed outside-support terms.
For rational data choose rational complements and $a_*$; the latter exists
by avoiding finitely many lines in a rational plane in $A_0$.
\end{proof}

\begin{theorem}[Openness on the secant variety]\label{thm:projectionopen}
Under the hypotheses of Lemma~\ref{lem:saturatedprojection}, let $r\ge2$
and suppose the paired Segre intersection is the $r$ displayed
\emph{reduced} points.  There is a Zariski-open neighborhood of $T$ in
the reduced affine $r$th secant variety of
$\operatorname{Seg}(\mathbb P(A)\times\mathbb P(B)\times\mathbb P(C))$
on which the same fixed projections certify rank $r$ and a unique minimal
complex decomposition.  The recovered paired points form a finite
\'{e}tale family of degree $r$ and span the common pair kernel.
After arbitrary enlargement of the ambient spaces the same conclusions
hold using the finite completion of Lemma~\ref{lem:projectionambient}.
No normality or smoothness of the secant variety is required.
\end{theorem}

\begin{proof}
First work in the given ambient spaces.  Denote the reduced affine secant
variety by $\Sigma$, and put $N=\dim A\dim C$ and $\rho_s=rk_s$.
It is integral, being the reduced closure of the image of the irreducible
space of ordered $r$-term decompositions.  The rank-one bound gives
$\rank Y_s(T')\le\rho_s$ throughout $\Sigma$.  Choose nonzero central
$\rho_s$-minors.  On their simultaneous nonvanishing chart the ranks equal
$\rho_s$, and elimination gives regular full left-kernel matrices $L_s(T')$
and a regular stacked pair map $G(T')$.

The ordered tuples with independent pairs form a nonempty open subset
of the tuple space, hence have dense image in $\Sigma$.  On this dense
image inside the chart, saturation puts their $r$ independent pairs in
$\ker G(T')$.  All $(N-r+1)$-minors of $G$ therefore vanish identically.
A central nonzero $(N-r)$-minor now gives, on a smaller neighborhood $U$,
a rank-$r$ kernel bundle $K$.  Moreover
$T'\in B\otimes K(T')$ for every $T'\in U$: for each $\beta\in B^*$
and $\omega$, applying $G_s$ to the $B$-contraction $T'(\beta)$ gives
$L_s(T')Y_s(T')(\omega\otimes\beta)=0$.

Let $\pi:Z\to U$ be the closed incidence subscheme of $\mathbb P(K)$
cut out by the paired Segre minors.  Its central fiber is reduced of
length $r$.  The morphism is projective, and every point of that fiber
is in its quasi-finite locus.  Remove from $U$ the closed image of the
complement of this locus.  The resulting morphism is proper and
quasi-finite, hence finite~\cite[Tag 02LS]{StacksProject}.
The geometric generic fiber has at least $r$ distinct points:
the dominant map from the open tuple locus above supplies a decomposition
with independent pairs after a field extension, and saturation places
those pairs in this fiber.

Set $R=\mathcal O_{U,T}$, a domain, and let $M$ be the finite $R$-algebra
of this incidence.  Since $M/\mathfrak m_RM\cong\bbC^r$, Nakayama gives
a surjection $R^r\twoheadrightarrow M$.  The generic-fiber observation
bounds the generic rank of $M$ below by $r$, and the surjection bounds it
above by $r$.  Its kernel has rank zero and is a submodule of a free module
over a domain, so it vanishes.  Hence $M$ is free of rank $r$; spreading
this isomorphism to a neighborhood makes $\pi$ finite locally free of
degree $r$.  The central trace pairing is nondegenerate because its algebra
is $\bbC^r$.  Invert its discriminant to make $\pi$ finite
\'{e}tale~\cite[Tag 0BJF]{StacksProject}.

Restriction of linear forms gives
$K^*\to\pi_*\mathcal O_Z(1)$ between rank-$r$ bundles; the target is
locally free because $\pi$ is finite locally free and $\mathcal O_Z(1)$
is invertible.  At $T$ this map is an isomorphism, since the displayed
points are independent.  Retaining its nonzero determinant makes all
fibers consist of exactly $r$ independent rank-one directions spanning
$K(T')$.  The inclusion $T'\in B\otimes K(T')$ supplies an $r$-term
decomposition.  Saturation gives rank at least $r$, so none of its terms
vanishes, and Lemma~\ref{lem:saturatedprojection} proves uniqueness.

For larger ambient spaces first apply Lemma~\ref{lem:projectionambient}.
The paired intersection in the fixed $S\subseteq A_0\otimes C_0$ is
unchanged as a scheme: restricting the enlarged matrix minors to $S$
gives exactly the old minors and zero equations.  The preceding proof
therefore applies in the larger secant variety as well.
\end{proof}

The reducedness hypothesis is supplied by Lemma~\ref{lem:pairedreduced}
for our subset certificates.  A mere count of rank-one directions does
not control the central scheme length in the freeness argument.
The restriction $r\ge2$ in the ambient completion is also material:
for $r=1$, targets with $d_s=1,p_s=0$ can have zero pair equations even
after enlarging $A$.  Rank-one uniqueness itself holds on the nonzero
locus of the first secant variety.

\paragraph{An explicit certificate for the fourth exceptional subset.}
The summand labels are
\[
 4,12,15,17,25,28,32,34,38,40.
\]
Use the $A$ basis $(a_4,a_{12},a_{15},a_{17},a_{25},a_{28},a_{32})$,
and the two maps $P_0,P_1:\bbC^7\to\bbC^5$ given by
\[
 P_0=\left[I_5\,\middle|\,
 \begin{matrix}-7&2\\-3&-5\\-7&4\\-5&3\\8&3\end{matrix}\right],
 \qquad
 P_1=\left[I_5\,\middle|\,
 \begin{matrix}-1&-7\\-2&8\\-1&5\\9&-6\\-9&-3\end{matrix}\right].
\]
Take $p_0=p_1=2$.  Both $70\times70$ Koszul matrices have exact rank
$60=10\binom42$.  Each projected horizontal pair map has size
$100\times35$ and rank $23$, so its individual kernel has dimension $12$.
After composition with $P_s\otimes\mathrm{id}_C$, their stack on the
original $49$ paired-product coordinates has size $200\times49$ and exact
rank $39$.  The ten displayed original $A$--$C$ product columns are
independent and annihilated by this stack; they therefore form its entire kernel.
The minimum left sides of \eqref{eq:pairedsylvester}, for support sizes
$2,\ldots,10$, are
\[
 2,3,3,4,4,4,4,5,4.
\]
These exact checks cover all $1013$ subsets of size at least two.
Lemma~\ref{lem:saturatedprojection} proves rank ten and global uniqueness
over $\bbC$ for the original tensor.  In particular, no assertion that either
individual projected kernel already equals the displayed span is needed.
Lemma~\ref{lem:pairedreduced} and Theorem~\ref{thm:projectionopen} also
give a certified open neighborhood in the rank-ten secant variety.
In the original $16$-dimensional $A$ space, ambient completion uses at most
$2(1+16-7)=20$ projections.  Pullback along a regular map from an
irreducible algebraic curve into this secant variety through the tensor gives a nonempty open
set, with finite complement for a finite-type curve; rational families
must first exclude their poles.  For the actual parameter family,
Corollary~\ref{cor:strong10curve} combines the uniform core theorem with
finite open conditions on the remaining supports to obtain generic strong
radius ten, without asserting a component-wide strong-ten bound.

% Full raw-coordinate input is checked before seven-dimensional compression.
% Primary certificate:
% reports/unique10_followup/xhigh/remaining_half/combined_projection.py
% Independent implementation, importing none of the primary construction:
% reports/unique10_followup/zero_corner/remaining_half/stacked_projection_independent.py

\section{Global uniqueness through eleven terms}\label{app:strong11}

This section uses the complete ten-set census of Proposition~\ref{prop:lp10}
and the strong-ten conclusion of Theorem~\ref{thm:strong10} as established
inputs. Its pointwise certificates concern the rational decomposition $D(2)$;
the final transfer argument establishes strong eleven on an open neighborhood
of $2$ along the actual parameter curve.
Write $d(I)=d_A(I)+d_B(I)+d_C(I)$ for its factor-span sum.

The next lemma reformulates \cite[Cor.~20]{LovitzPetrov}: in three modes
its proper-subpartition threshold is $n+r\le d([n])-2$.
Taking $r=n-1$ gives the rank threshold; taking $r=n$ and using
proper-subset uniqueness on the resulting blocks gives the uniqueness
threshold. These thresholds are inherited, not new splitting results.
We include the connected-component proof to expose the hypotheses used
by the finite certificates.

\begin{lemma}[Lovitz--Petrov splitting corollary]\label{lem:splitting11}
Let $x_1,\ldots,x_n$ be independent nonzero three-factor product tensors,
and suppose every nonempty proper displayed subset is rank-minimal.
If $d([n])\ge2n+1$, their sum has rank $n$.
If all proper displayed subsets are also identifiable and
$d([n])\ge2n+2$, its displayed decomposition is globally unique.
\end{lemma}
\begin{proof}
Compare with a minimal decomposition $\sum_{j=1}^r y_j$, $r\le n$, and
form the labeled vector multiset $E=(x_1,\ldots,x_n,-y_1,\ldots,-y_r)$.
The spans of its vector-matroid connected components form a direct sum,
so the zero sum restricts to a zero sum in each component.
Independence of the $x_i$ and minimality of the competitor exclude
components containing only one side. If $E$ is disconnected, each component
therefore contains a nonempty proper displayed subset. The corresponding
competing subdecomposition is minimal (otherwise the entire competitor
shortens), and hence has the same cardinality. Thus $r=n$; with proper-subset
uniqueness the tensor multisets also agree component by component.
A shorter competitor, or a distinct minimal competitor in the uniqueness
case, must consequently give connected $E$.
The contrapositive of the splitting theorem of
\cite[Thm.~4]{LovitzPetrov} gives
\[
 d([n])-2\le d_A(E)+d_B(E)+d_C(E)-2
 \le\dim\spn E\le n+r-1.
\]
Thus $r\ge d([n])-n-1$. The first threshold excludes $r<n$;
after rank $n$ is known, the second excludes distinct minimal competitors.
\end{proof}

The three paired-product matrices of $D(2)$ have column rank $48$,
certified by nonzero minors modulo $1000003$. This directly identifies
the intrinsic supports of every displayed subset with its factor spans;
it does not assume minimality of the full $48$-term decomposition.
In any minimal competitor the paired columns are independent, since a
dependence permits absorption of one remaining-mode factor into the others
and shortens the decomposition. Its factors therefore belong to those
same intrinsic supports. All subsequent coordinate compressions are checked
to be injective on the relevant spans, using all original coordinates.

\paragraph{An exhaustive reduction to the core.}
Use the sixteen-element set
\[
 X=\{4,9,12,15,17,19,22,25,28,32,34,35,38,40,44,48\}.
\]
Its factor spans have dimension eight in each mode. The rational quotient
maps in \texttt{quotient\_certificate.json} have exactly those kernels.
Each outside factor maps to a nonzero multiple of one of eight coordinate
directions, each occurring four times. For $I=J\mathbin{\dot\cup}K$,
$J\subseteq X$, $K\subseteq X^c$, let $t(K)$ be the sum of the numbers
of quotient directions hit in the three modes. Dimension in the quotient
and in its kernel gives
\[
 d(I)\ge d(J)+t(K),\qquad
 d(J)\ge\mu_{|J|},\qquad
 (\mu_0,\ldots,\mu_{11})=(0,3,6,9,11,14,15,17,18,20,21,23).
\]
The latter inequalities are certified by modular nonzero minors for every
core subset, so remain lower bounds over $\bbC$.

Suppose $|I|=11$, $K\ne\varnothing$, and $d(I)\le23$.
If $|K|=1$, the bound $21+3=24$ is a contradiction.
If $|K|\ge2$, deleting any outside index leaves an outside-containing
ten-set, whose factor-span sum is at least $22$ by
Proposition~\ref{prop:lp10}. An index singleton in a quotient direction
is a coloop in that mode. If it were singleton in two modes, deleting it
would instead give sum at most $21$. Thus each outside index of a
candidate is singleton in at most one quotient mode.

Enumerate increasing outside subsets $K$ to depth eleven. A subtree may
be stopped when the sum of the three factor ranks over $\mathbb F_3$,
taken as an ordinary integer sum, reaches $24$: this lower bound is
monotone. At each remaining node, complete it
by all $J\subseteq X$ of size $11-|K|$ only if the singleton condition
and $t(K)+\mu_{11-|K|}\le23$ hold. These two gates skip completions
at that node only, \emph{not} its descendants. Every possible
counterexample must survive these gates. The complete counts are
\[
\begin{array}{c|r|r|r}
 |K|&\text{surviving }K&\text{completed supports}&d(I)<24\\\hline
 6&32&139776&0\\
 7&6720&12230400&0\\
 8&4592&2571520&0
\end{array}
\]
All other outside sizes give no completions. The $14941696$ completed
checks all have sum of the three factor ranks over $\mathbb F_3$, taken as
an ordinary integer sum, at least $24$. Two GF$(3)$ elimination
implementations give the same counts. Thus all outside-containing eleven-sets
have $d(I)\ge24$, and Lemma~\ref{lem:splitting11} proves uniqueness.

Among the $\binom{16}{11}=4368$ core sets, the exact profiles
$(7,8,8),(8,7,8),(8,8,7),(8,8,8)$ occur respectively
$48,304,304,3712$ times. Modular lower bounds and rational rank-seven
checks establish these equalities. In particular every eleven-set has
$d(I)\ge23$, already proving rank eleven. Only the $656$ core sets
with sum $23$ need additional uniqueness certificates.

\paragraph{The $624$ saturated certificates.}
In injective eight-dimensional core coordinates, the stored projections
$P:A\to\bbC^5$, with exterior degree $2$, give $80\times80$ Koszul
matrices of rank $66=11\binom42$. For each of $624$ supports, two or
three such maps give a stacked pair map of rank $53$ on the original
$64$-dimensional paired space. Its kernel is the span of the eleven
independent displayed pairs. There are $336$ two-map and $288$ three-map
certificates, totaling $1536$ maps. Over every core subset of sizes
$2,\ldots,11$, in every pairing, the minima in
\eqref{eq:pairedsylvester} are
\[
 2,3,3,4,4,4,4,4,4,4.
\]
Sylvester's inequality therefore excludes all other \emph{complex}
rank-one directions in each paired span.
Lemma~\ref{lem:saturatedprojection} applies: saturation puts every
alternative pair in the common kernel, and coefficient comparison in the
independent displayed pair basis recovers its remaining factor.

For completeness, the modular kernel calculation lifts to characteristic
zero. Rank $66$ modulo $p=1000003$ meets the universal upper bound $66$.
Elimination at a pivot minor nonzero modulo $p$ takes place over the
localized coefficient ring and gives a rational left-kernel basis reducing
to the computed one. The nonzero modular $53$-minor thus belongs to a
rational stacked pair map. Its eleven known independent kernel columns
give the matching upper bound $53$. These are positive-minor certificates,
not finite-field point counts. An independently indexed Kronecker-product
construction at prime $65521$ verifies the same conditions.

The rank step in the following one-arm reduction is the slice-substitution
bound of \cite[Lemmas~2--3]{Shitov}. The zero-shear condition adds the
purity needed for the uniqueness conclusion.
\begin{lemma}[One-arm completion]\label{lem:onearm11}
Let $C=C_0\oplus C_1$, $\dim C_1=c$, and $T=H+U$, with
$H\in A\otimes B\otimes C_0$ of rank $h$ and
$U\in A\otimes B\otimes C_1$ of rank $c$ and full $C_1$ support.
For every $R:C_1\to C_0$ suppose
\[
 \rank H(R)\ge h,\qquad
 \rank H(R)\le h\ \Longrightarrow\ R=0,
 \quad H(R)=H+(\mathrm{id}\otimes\mathrm{id}\otimes R)U.
\]
Then $T$ has rank $h+c$, every minimal decomposition is pure in the two
$C$ blocks, and uniqueness of $H,U$ implies uniqueness of $T$.
\end{lemma}
\begin{proof}
Apply Theorem~\ref{thm:zerocorner} with $B_1=0$, $L=0$, $q=h$, and
the stated space $C_1$.  Its rank clause gives $\rank T=h+c$.  The
zero-shear clause gives purity and the stated uniqueness conclusion.
\end{proof}

\begin{lemma}[Fixed-strip inequality]\label{lem:fixedstrip11}
Let $A=A_F\oplus A_B$, $B=B_F\oplus B_B$, with
$\dim A_F=\dim B_F=4$. Suppose the four $B_F$ slices of $W$ are
$e_s\otimes c_s\in A_F\otimes C$, each $c_s\ne0$.
If $Z$ is its projection onto $A_B\otimes B_B\otimes C$, then
$\rank W\ge4+\rank Z$.
\end{lemma}
\begin{proof}
The four slices are independent. Choose four competing $B$ factors whose
$B_F$ projections form a basis and kill them by a retraction onto $B_B$.
Subsequent projection onto $A_B$ kills the entire retraction error from
the fixed strip and leaves exactly $Z$, with at least four terms removed.
\end{proof}

\paragraph{The $32$ completion certificates.}
For each remaining support the rational certificate specifies a mode
permutation and $I=F\mathbin{\dot\cup}B\mathbin{\dot\cup}U$ with
sizes $4,4,3$. The eight $F,B$ factors are bases in the first two modes;
their third-mode span $C_0$ has dimension four, complemented by the
three $U$ factors. The arm's second-mode coordinates vanish in the
$F$ strip. These statements are checked by exact reconstruction from
all sixteen coordinates. Put $H=T_{F\cup B}$; the strong-ten base
gives rank and uniqueness of $H$ and $T_U$.

In the indicated bases the bottom tensor of $H(R)$ has slices
\[
 M_s(R)=E_{ss}+\sum_{j=1}^3u_{sj}v_jr_j^{\mathsf T},
 \qquad R=[r_1\ r_2\ r_3]\in\bbC^{4\times3}.
\]
Here $u,v$ are the bottom four coordinates of the first two arm factors.
Each certificate gives $\rank u=2$ and
$\alpha\in\ker u^{\mathsf T}$ with no zero coordinate. The slice
$D_0=\sum_s\alpha_sM_s=\operatorname{diag}(\alpha)$ is constant and
invertible for \emph{every} $R$. Thus the bottom tensor has rank at
least four and Lemma~\ref{lem:fixedstrip11} gives $\rank H(R)\ge8$.
If $\rank H(R)\le8$, the bottom tensor has rank four. Its invertible
slice makes both square matrix-factor families invertible, so necessarily
\begin{equation}\label{eq:completion11}
 M_sD_0^{-1}M_t-M_tD_0^{-1}M_s=0.
\end{equation}
For a basis of $\ker u^{\mathsf T}$, the additional constant slices
$D_\beta=\operatorname{diag}(\beta)$ give linear equations
$M_sD_0^{-1}D_\beta-D_\beta D_0^{-1}M_s=0$ in the twelve entries of $R$.
For each support their coefficient matrix has exact rank eight. The
certificate records its entire four-dimensional kernel. Substitution into
\eqref{eq:completion11} gives four selected entries that are homogeneous
linear polynomials in those four parameters with nonsingular coefficient
matrix. Checking polynomial equality includes verifying all quadratic
coefficients vanish; this is not a linearization. Consequently $R=0$.
Lemma~\ref{lem:onearm11} now proves global uniqueness.

Here is one explicit certificate. Take
$F=(9,25,28,35)$, $B=(12,19,22,32)$, $U=(4,17,34)$ in the original
mode order. Use $F,B$ as the first two bases and $B,U$ as the third.
Then
\[
 u=\frac12\begin{pmatrix}-1&1&-1\\-1&-1&1\\-1&-1&1\\-1&1&-1\end{pmatrix},
 \qquad
 v=\frac1{16}\begin{pmatrix}1&4&-4\\1&4&4\\1&-4&-4\\-1&4&-4\end{pmatrix}.
\]
Choose $\alpha=(-1,-1,1,1)$ and $\beta=(0,-1,1,0)$.
The full linear solution is
\[
 R=\begin{pmatrix}0&\theta_1&\theta_1\\0&-\theta_2&\theta_2\\
 0&-\theta_3&\theta_3\\0&\theta_4&\theta_4\end{pmatrix}.
\]
For $C=M_1D_0^{-1}M_2-M_2D_0^{-1}M_1$, its entries
$(C_{14},C_{23},C_{32},C_{41})$ are
$(\theta_4,-\theta_3,-\theta_2,-\theta_1)/4$, with coefficient determinant
$-1/256$. Hence all parameters vanish. The other $31$ supports have
their own certificates; no unproved symmetry extrapolation is used.

\begin{theorem}[Strong rigidity through eleven]\label{thm:strong11}
Over $\bbC$, $D(2)$ is strongly $11$-locally rigid. In particular
$\srho(D(2))\ge11$ and $\rrho(D(2))\ge11$.
\end{theorem}
\begin{proof}
The strong-ten base handles proper subsets. Splitting handles every
outside-containing eleven-set and $3712$ core eleven-sets. The remaining
$656$ supports are exactly the disjoint union of the $624$ projection and
$32$ completion certificate sets, verified by equality of index sets.
The preceding arguments establish rank and global uniqueness for all of them.
\end{proof}

\begin{corollary}[Generic rank radius at least eleven]\label{cor:rank11generic}
There is a Zariski-open neighborhood of $D(2)$ in
$V_{48}(\M{4,4,4})$ on which the rank radius is at least $11$.
It is dense in every irreducible component through $D(2)$, and
$\rrho(D(t))\ge11$ for all but finitely many $t\in\bbC^\times$.
\end{corollary}
\begin{proof}
Let $W_{10}$ be the open neighborhood supplied by
Corollary~\ref{cor:rank10generic}, on which all subsets of size at most
ten are minimal. Intersect it with the open locus of rank $48$ for the
paired-product matrices and, for every eleven-set $I$, the locus $d(I)\ge23$.
Each latter locus is a finite union of intersections of open determinantal
conditions. There are finitely many $I$, and the census and paired minors
show that the resulting open set $W_{11}$ contains $D(2)$.
At every point of $W_{11}$ the displayed tensors are independent, proper
subsets are minimal, and Lemma~\ref{lem:splitting11} proves rank eleven.
Nonempty openness gives the component assertion. The actual family has
coordinates in $\mathbb Z[1/2,t,t^{-1}]$ and nonzero factors on
$\bbC^\times$, and its tensor identity gives a regular morphism into
$V_{48}$. The inverse image of $W_{11}$ contains $2$ and is open in
$\bbC^\times$, hence has finite complement.
\end{proof}

\begin{corollary}[Generic strong eleven on the curve]\label{cor:strong11curve}
There is a Zariski-open neighborhood of $2$ in $\bbC^\times$ on which
$D(\tau)$ is strongly $11$-locally rigid. In particular
$\srho(D(\tau))\ge11$ for all but finitely many nonzero complex parameters.
\end{corollary}
\begin{proof}
For every non-core support of size at most ten, retain the finite factor
minors that witness the Lovitz--Petrov inequalities at $t=2$; these are
the nonexceptional supports in Proposition~\ref{prop:lp10}, together with
the already certified lower levels. Retain also the nonzero paired-product
minors at $2$. Each is a determinantal open condition on the parameter
curve. Their finite intersection $U_{\le10}$ contains $2$. On it,
Theorem~\ref{thm:lp} proves rank-minimality and global uniqueness for every
non-core support through level ten. Every core support through level eleven,
including the exceptional ten-supports, is instead covered directly for every
nonzero parameter by Theorem~\ref{thm:uniformcore11}.

For each outside-containing eleven-support retain factor minors witnessing
$d(I)\ge24$ at $2$, as established by the complete mixed-support census.
There are finitely many such supports. Their intersection with $U_{\le10}$
still contains $2$; on it every proper face is rank-minimal and identifiable,
and Lemma~\ref{lem:splitting11} proves global uniqueness. Thus this finite
intersection certifies every displayed support through level eleven.
A nonempty open subset of the irreducible curve $\bbC^\times$ has
finite complement.
\end{proof}

The rank transfer in Corollary~\ref{cor:rank11generic} uses open
sufficient matrix-rank conditions, not semicontinuity of tensor rank.
The strong-eleven transfer is restricted to the actual parameter curve;
no level-eleven strong bound on whole components of the decomposition
variety follows here. The ambient-component strong bound remains nine,
and the certified border bound remains eight.

\paragraph{Reproducibility scope.}
The new certificate package reconstructs the full matrix-multiplication
identity, quotient maps and core ranks from \texttt{family\_t2.json},
replays the mixed census with two GF$(3)$ implementations, verifies all
saturated certificates with separate constructions at two primes, and
checks all completion polynomial identities over $\bbQ$.
The command \texttt{python verify.py --independent} checks equality of
the coverage sets as well. Its level-ten inputs are established results:
matching the supplied bad-ten list does not independently prove that list
complete, nor do representative ten-set notes replace their all-case checks.
The new replay and those base checks are separate evidence dependencies.

\section{A sharp strong radius and a uniform core}\label{app:sharpuniform}

Put
\[
 G_1=(4,17,34,40),\quad G_2=(9,25,28,35),\quad
 G_3=(12,19,22,32),\quad G_4=(15,38,44,48).
\]

\begin{proposition}[An exact twelve-term replacement]\label{prop:strongsharp}
The strong radius of $D(2)$ is exactly $11$. The least number of its
summands changed by a distinct $48$-term decomposition is exactly $12$.
\end{proposition}
\begin{proof}
The lower bound is Theorem~\ref{thm:strong11}. Here is a rational upper
certificate, independent of any assertion of rank-twelve minimality.
Put $F=(9,25,28,35)$, $G=(12,19,22,32)$ and $U=(4,17,34,40)$.
Use the original factors of $F,G$ as ordered bases in the first two
modes and those of $G,U$ in the third. All three bases have rank eight.
Write $c_f$ for the first four coordinates of the third factor of $F_f$,
and $a_j,b_j$ for the eight coordinates of the first two factors of $U_j$;
write $\bar a_j,\bar b_j$ for their bottom four coordinates.
The third factor of $U_j$ is $(0,e_j)$ and $b_j=(0,\bar b_j)$.
Define matrices by columns, $R=[r_j]$, $A'=[a'_j]$, $B'=[b'_j]$:
\[
\begin{gathered}
R=\frac18\begin{pmatrix}12&12&12&-3\\12&-12&12&3\\
12&12&-12&3\\-12&12&12&3\end{pmatrix},\quad
A'=\frac18\begin{pmatrix}9&-3&-3&1\\-3&9&1&-3\\
-3&1&9&-3\\1&-3&-3&9\end{pmatrix},\\
B'=\frac12\begin{pmatrix}3&0&0&1\\0&3&-1&0\\
0&-1&3&0\\1&0&0&3\end{pmatrix},\qquad C'=B'/2.
\end{gathered}
\]
Exact multiplication gives
\[
 \sum_{s=1}^4 e_s^{\otimes3}
 +\sum_{j=1}^4\bar a_j\otimes\bar b_j\otimes r_j
 =\sum_{j=1}^4 a'_j\otimes b'_j\otimes c'_j.
\]
For $j_f=(3,1,4,2)_f$ and $\lambda_f=(3,-3,-3/4,-3)_f$,
one has $r_{j_f}=\lambda_fc_f$ and the top four coordinates of
$a_{j_f}$ are $e_f$. The following twelve triples therefore have the
same tensor sum as $F\cup G\cup U$:
\[
\begin{split}
 &(e_f,\ e_f+\lambda_f(0,\bar b_{j_f}),\ (c_f,0))\quad(1\le f\le4),\\
 &((0,a'_j),(0,b'_j),(c'_j,0))\quad(1\le j\le4),\\
 &(a_j,b_j,(-r_j,e_j))\quad(1\le j\le4).
\end{split}
\]
Here the $e_f$ in the first two modes denotes the corresponding vector
in eight dimensions. Indeed the bottom correction is the displayed
matrix identity, and its top correction cancels using
$r_{j_f}=\lambda_fc_f$. All twelve products are nonzero and mutually
distinct, and none equals any of the original forty-eight products.
These finite rational assertions, including the complete reconstruction
in sixteen coordinates per mode, are checked by
\path{strong12_obstruction_audit/verify_explicit.py}; the separate
\path{verify_export.py} checks the exported factors directly in all
$4096$ tensor coordinates. The supplementary data map locates both checkers.
The checked new terms do not match
any removed summand.  Since the original terms are pairwise distinct, the
unchanged $36$ terms are exactly the common tensor-summand multiset: a new
term that happened to equal an unchanged term would only add a second copy
and could not increase that common multiplicity.  Retaining the other $36$
terms therefore proves the upper bound and the exact equal-length distance.
No shorter decomposition of this twelve-term partial sum is asserted or
needed.
\end{proof}

\begin{remark}[Tensor-summand versus input-product reuse]
The replacement reuses $36$ complete tensor summands.  In the usual
first-two-modes convention it reuses $40$ input-product directions: the
four directions indexed by $4,17,34,40$ recur with changed output factors.
Thus this example has tensor-summand distance $12$ and input-product
distance $8$.  It makes no minimality assertion for the latter metric.
The exact comparison is checked by
\path{evidence/metrics/check_product_distance.py}.
\end{remark}

\begin{remark}[An inequivalent replacement]\label{rem:replacementorbit}
At $t=2$, the twelve new summands have factor-matrix rank triples
$(2,3,2)$, $(2,3,3)$ and $(2,2,3)$, four of each. Every factor matrix
of the original $D(t)$ has rank at most two for every $t\ne0$.
Matrix-multiplication isotropy acts by invertible left/right matrix maps
and possible transposition and mode permutation (the supplementary
proposition on isotropy invariance); it preserves these ranks,
as do factor gauges and term permutations. Hence the replacement is not
equivalent to any member of the original family under these transformations.
The sharp distance twelve is therefore attained in a different isotropy
orbit, not merely by another gauge/permutation representative. This does
not classify the replacement among all previously known algorithms.
Exact certificates are included in the mathematical supplement's data map.
\end{remark}

\begin{lemma}[Coefficient stability]\label{lem:weightstable}
Every nonzero coefficient reweighting of a decomposition certified by
Lemma~\ref{lem:saturatedprojection} has the same rank and a unique minimal
decomposition. With zero coefficients allowed, the rank is the support
size and the unique decomposition is the displayed subdecomposition.
\end{lemma}
\begin{proof}
For a selected flattening write its summand matrices as $M_i=U_iV_i$,
each of rank $k$. Saturation to rank $rk$ forces $[U_1\ \cdots\ U_r]$
to be injective and the stacked $V_i$ to be surjective. Inserting
$\operatorname{diag}(\lambda_iI_k)$ for nonzero $\lambda_i$ preserves
rank and image. Thus all left kernels and pulled-back pair constraints
are unchanged, as are the displayed pair span and its rank-one directions.
The same saturated-image criterion applies. For a smaller support, fill
the missing coefficients with nonzero values. A shorter or different
equal-length replacement would contradict uniqueness of the full weighted
sum after cancellation of the common complementary tensor multiset.
\end{proof}

\begin{lemma}[Laurent normal form]\label{lem:coreweight}
Let $X$ be the sixteen-index core of Section~\ref{app:strong11},
$Q=\{15,38,44,48\}$ and $\eta=t/2$. There are ambient invertible
Laurent-polynomial maps $\Phi_B,\Phi_C$ such that, for every $i\in X$,
\[
a_i(t)=a_i(2),\qquad \Phi_Bb_i(t)=b_i(2),\qquad
\Phi_Cc_i(t)=\eta^{\mathbf1_Q(i)}c_i(2).
\]
Their determinants are $\eta^{-4}$ and $\eta^4$, respectively.
\end{lemma}
\begin{proof}[Exact certificate]
For each core factor matrix $F_m(t)$ choose basis labels $J_m$ and
coordinate rows $R_m$ as follows, with all indices one-based:
\[
\begin{array}{c|l|l|l}
m&J_m&R_m&\det F_m[R_m,J_m]\\\hline
A&4,9,12,15,17,19,22,25&1,2,3,5,7,9,10,13&8\\
B&4,9,12,15,17,22,25,28&1,2,3,5,6,8,9,10&-1024t^4\\
C&4,9,12,15,17,19,22,34&1,2,3,5,6,7,9,11&(262144t^3)^{-1}
\end{array}
\]
Put $E_m=F_m[:,J_m]$, $K_m=E_m[R_m,:]^{-1}F_m[R_m,:]$ and
$\Lambda=\operatorname{diag}(\eta^{\mathbf1_Q(i)})_{i\in X}$.
The exact rational-function identities are
\[
F_m=E_mK_m,\qquad K_B(t)=K_B(2),\qquad
\Lambda_{J_C}K_C(t)=K_C(2)\Lambda.
\]
Complete $E_m$ by the standard coordinate vectors outside $R_m$ to a
square frame $N_m$. Then take
\[
\Phi_B=N_B(2)N_B(t)^{-1},\qquad
\Phi_C=N_C(2)\operatorname{diag}(\Lambda_{J_C},I_8)N_C(t)^{-1}.
\]
The listed minors are Laurent units, so these maps and inverses are
defined for every nonzero complex parameter. Their determinants and
factor action follow from the identities. The supplementary
\path{uniform_normal_form_audit/independent_check.py} verifies all
identities and inverses over $\bbQ(t)$, and checks all $2304$ input
coordinates against the original symbolic family.
\end{proof}

\begin{theorem}[Uniform strong eleven on the core]\label{thm:uniformcore11}
For every $t\in\bbC^\times$, each partial sum of at most eleven core
terms has rank equal to its size and a globally unique minimal complex
decomposition.
\end{theorem}
\begin{proof}
The normal form preserves every restricted factor matroid and paired
independence condition. Thus all core Lovitz--Petrov certificates through
ten persist. The sixteen exceptional saturated ten-sets persist by
Lemma~\ref{lem:weightstable}. For each of the sixteen zero-corner
ten-sets, $Q$ meets a whole block or is disjoint from it. In the notation
$H,L,U$ of Section~\ref{app:zerocorner}, the weighted sum is
$hH+\ell L+uU$, with all weights nonzero. Scale $B_0,B_1$ by
$h^{-1},\ell^{-1}$ and $C_0,C_1$ by $1,h/u$; this sends it termwise
to its identifiable specialization. Hence the core is strongly ten-rigid
at every nonzero parameter.

The $3712$ core eleven-sets with factor sum $24$ now satisfy
Lemma~\ref{lem:splitting11}. The $624$ saturated cases persist by
coefficient stability. For each of the remaining $32$ cases retain the
partition $F,B,U$ and mode permutation of Section~\ref{app:strong11}.
The weights are constant on each partition block, say $f,b,u$.
In the second mode, $F$ spans $B_F$ while $B,U$ lie in its complementary
block $B_B$; in the third, $F,B$ lie in $C_0$ and $U$ in $C_1$.
Scaling $B_F,B_B$ by $f^{-1},b^{-1}$ and $C_0,C_1$ by $1,b/u$
sends the weighted decomposition termwise to the specialization.
Its global uniqueness was proved in Theorem~\ref{thm:strong11}.
The supplementary independent checker verifies this block geometry
and the whole-block condition for all $32+16$ partitions. This argument
uses the established exhaustive census and pointwise certificates; it
does not replace them by a representative or by a symmetry assumption.
\end{proof}

\begin{corollary}[Uniform sharpness on the core]\label{cor:uniformcoresharp}
For every $t\in\bbC^\times$, the restriction of $D(t)$ to the
sixteen-term core $X$ has strong radius exactly $11$.  It has a distinct
sixteen-term decomposition at tensor-summand distance $12$.
\end{corollary}
\begin{proof}
The lower bound is Theorem~\ref{thm:uniformcore11}.  The replacement of
Proposition~\ref{prop:strongsharp} removes
$G_1\cup G_2\cup G_3$, which is disjoint from the weighted tetrad $G_4$.
Lemma~\ref{lem:coreweight} therefore carries both its removed summands and
its twelve replacement summands termwise to their checked $t=2$ values.
Apply the inverse normal form and retain $G_4$.  The checked new terms do
not match a removed summand; the four retained terms exhaust the common
core multiset, by the same multiplicity argument as in
Proposition~\ref{prop:strongsharp}.  Thus the resulting distinct core
decomposition has distance $12$, proving the upper bound on the strong
radius and hence equality.
\end{proof}

This is a uniform statement about the core, not about all mixed supports
of the forty-eight-term family. The new normal-form and replacement
checks are separate supplementary dependencies, outside the unchanged
historical \texttt{v1.3} archive and its fixed-parameter verifier.

\begin{remark}[The remaining ambient question]
On an ambient neighborhood retaining the rank-eleven, strong-nine,
factor-minor and saturated-projection certificates, strong eleven is
equivalent to uniqueness of the $32$ residual one-arm eleven-sums.
Indeed their ten-subsets cover the sixteen zero-corner exceptions;
the other sixteen exceptional ten-sets have projection-certified
eleven-extensions. Uniqueness descends to subsets of a minimal
decomposition by replacement and cancellation. Thus these extensions
first give strong ten; splitting then handles the ordinary eleven-sets,
and the $624$ projections and $32$ hypotheses finish eleven.
The supplementary \path{ambient_reduction_audit/} verifies this incidence
coverage and full ambient differential ranks $460$ for the sixteen
zero-corner ten-sums and $506$ for the thirty-two one-arm eleven-sums,
in gauge-fixed rank-one coordinates ($46$ per summand).

For the latter sums, the injective differential and pointwise global
uniqueness also exclude affine collision of distinct minimal competitors
at the displayed fiber. In a rank-one chart choose $506$ output coordinates
with nonsingular differential. The polynomial identity
$F(v)-F(q)=H(q,v)(v-q)$ has $H(q,q)$ invertible, so equality near
the diagonal forces $v=q$. Apply this at every permutation fiber.
Any affine limiting tuple has eleven rank-one-or-zero terms; rank eleven
excludes zero terms and uniqueness leaves only those fibers. Competitors
escaping to infinity in tensor-summand coordinates remain possible.
No boundary-emptiness or geometric-component elimination certificate is
provided, so this reduction does not improve the ambient strong bound.
\end{remark}

\section{Pointwise rank rigidity through twelve}\label{app:rank12}

The results of this section concern the specified rational decomposition
$D(2)$. We use Theorem~\ref{thm:strong11}, the established bound
$d(I)\ge24$ for every outside-containing eleven-set, and the earlier
complete bound $d(I)\ge22$ for every outside-containing ten-set.
The new certificates do not replace the verification of these inputs.
Lemma~\ref{lem:splitting11} gives rank twelve whenever $d(I)\ge25$;
the threshold $d(I)\ge26$ additionally gives global uniqueness.

\paragraph{Mixed supports.}
Retain the sixteen-element core $X$ from Section~\ref{app:strong11}.
For $I=J\sqcup K$, where $J\subseteq X$ and $K\subseteq X^c$,
let $t(K)$ be the number of occupied quotient directions, summed over
the three certified quotient maps. The verified core lower bounds give
\[
 d(I)\ge\mu_{|J|}+t(K),\qquad
 (\mu_0,\ldots,\mu_{12})=(0,3,6,9,11,14,15,17,18,20,21,23,23).
\]
If $|K|=1$, this is at least $26$. If $|K|\ge2$ and $d(I)\le25$,
no outside index can occupy a singleton quotient class in two modes:
deleting it would leave a mixed eleven-set of dimension sum at least
$24$. If $|K|\ge3$, deleting any pair leaves a mixed ten-set of
dimension sum at least $22$. Thus at most three occupied quotient
classes, summed over modes, may be entirely removed by that pair.
Each removed class supplies an independent lost quotient direction,
which proves this second deletion gate.

The exhaustive traversal enumerates increasing outside subsets to depth
twelve. Its only subtree pruning criterion is that the sum of the three
factor ranks over $\mathbb F_3$, taken as an ordinary integer sum, is already
at least $26$, a monotone lower bound in characteristic zero. The deletion
gates and the bound
$\mu_{12-|K|}+t(K)\ge26$ skip completions only, not descendants.
Every remaining core completion is tested in all sixteen coordinates,
first modulo $3$ and then, if needed, modulo $1000003$.
Two arithmetic implementations perform $218180816$ completed checks
each and return the same exact $52$-support residue. Rational checks
give $20$ profiles $(8,8,8)$, $16$ profiles $(8,9,8)$, and $16$ profiles
$(8,8,9)$. The latter $32$ have rank twelve by the splitting threshold.

\paragraph{Explicit lower-bound certificates.}
After coordinate projections to eight-dimensional spaces and a stored
projection of one mode to dimension five, use the Koszul map
\[
 Y(a\otimes b\otimes c)(\omega\otimes\beta)
   =(\Pi a\wedge\omega)\otimes c\,\beta(b),
 \qquad \omega\in\Lambda^2\mathbb C^5.
\]
Its matrix is $80\times80$, and a product tensor contributes rank at
most $\binom42=6$. A nonzero $67$-minor therefore proves border rank
at least twelve. The certificate lists explicit row and column indices
for $1836$ such minors; each is nonzero modulo both $1000003$ and
$65521$. A nonzero reduction at either prime proves nonvanishing over
$\mathbb Q$ and $\mathbb C$. The displayed twelve terms give the
matching upper bound. Coverage is as follows:
\[
\begin{array}{c|r|l}
 \text{support class}&\text{number}&\text{rank-twelve certificate}\\\hline
 \text{core}&1816&67\text{-minor}\\
 \text{core}&4&\text{uniform completion bound below}\\
 \text{mixed residue},\ d=24&20&67\text{-minor}\\
 \text{mixed residue},\ d=25&32&\text{splitting lemma}\\
 \text{other mixed supports}& &d\ge26
\end{array}
\]
The verifier checks exact support-set equality: the first two rows
partition all $\binom{16}{12}=1820$ core supports, and the next two
partition the mixed residue.

\paragraph{The four exceptional core supports.}
Put
\[
 G_1=(4,17,34,40),\quad G_2=(9,25,28,35),\quad
 G_3=(12,19,22,32),\quad G_4=(15,38,44,48).
\]
The exceptional supports are the unions of three tetrads. The exact
charts use the following mode order and groups:
\[
\begin{array}{c|c|c|c|c|c}
 I&\text{mode order}&F&B&U&\alpha\\\hline
 G_1\cup G_2\cup G_3&(A,B,C)&G_2&G_3&G_1&(-1,-1,1,1)\\
 G_1\cup G_2\cup G_4&(A,C,B)&G_2&G_4&G_1&(-1,1,1,1)\\
 G_1\cup G_3\cup G_4&(A,C,B)&G_3&G_1&G_4&(-1,1,1,1)\\
 G_2\cup G_3\cup G_4&(A,B,C)&G_3&G_2&G_4&(1,-1,1,1)
\end{array}
\]
Write $H=T_{F\cup B}$ and name the reordered modes $A,B,C$.
The eight factors of $H$ are bases in its first two modes, ordered
$F,B$; hence $\rank H=8$. The third factors of $F,B$ lie in $C_0$, and
those of $B,U$ are respective bases of the complementary four-spaces
$C_0,C_1$. All second-mode factors of $U$ have zero
top block. Let $u_j,v_j$ be the bottom first- and second-mode arm
coordinates. The certificates reconstruct these identities in all
sixteen original coordinates and verify $u^T\alpha=0$ with every
entry of $\alpha$ nonzero.

For any complex linear map $R:C_1\to C_0$, define
\[
 H(R)=H+(\mathrm{id}\otimes\mathrm{id}\otimes R)T_U,
 \qquad
 Z(R)=\sum_{s=1}^4e_s\otimes e_s\otimes e_s
          +\sum_{j=1}^4u_j\otimes v_j\otimes r_j.
\]
The latter is the bottom first/second block of $H(R)$. Its contraction
by $\alpha$ in the first mode is the constant invertible matrix
$\operatorname{diag}(\alpha)$, so $\operatorname{rank}Z(R)\ge4$.
In any $r$-term decomposition of $H(R)$, the second-mode projections
to the top four-space span that space. Choose four independent such
projections and retract the second mode onto its bottom block,
identically on that block, killing the chosen factors. Projecting the
first mode onto its bottom block kills the fixed strip and its
retraction error. What remains is precisely $Z(R)$, represented by at
most $r-4$ terms. Consequently
\[
 \operatorname{rank}H(R)\ge4+\operatorname{rank}Z(R)\ge8
 \quad\text{for every complex }R.
\]
The arm has rank four and full $C_1$-support.  The rank clause of
Theorem~\ref{thm:zerocorner}, with $B_1=0$, $q=8$, and $c=4$, therefore
gives $\rank T_I=12$.  This is a uniform rank argument, not a uniqueness
assertion for completions and not a border-rank subtraction argument.

\begin{theorem}[Pointwise rank rigidity through twelve]\label{thm:rank12}
Every displayed partial sum of at most twelve terms of $D(2)$ has
tensor rank equal to its number of terms. Hence $\rrho(D(2))\ge12$.
\end{theorem}
\begin{proof}
Theorem~\ref{thm:strong11} settles sizes at most eleven. The exhaustive
mixed reduction, the explicit minors, and the four uniform completion
arguments settle every twelve-support.
\end{proof}

Combined with Theorem~\ref{thm:strong11}, the exact twelve-term replacement in
Proposition~\ref{prop:strongsharp} yields
$\srho(D(2))=11$ and a rational $48$-term replacement
of the full decomposition. Every strictly shorter decomposition must
discard at least thirteen original tensor summands: cancellation of
the shared terms would otherwise contradict Theorem~\ref{thm:rank12}.
The parameter $z$ in the separate symbolic replacement family varies
decompositions of the fixed partial sum at original parameter $t=2$.
The supplied symbolic certificate verifies a nonconstant rational curve
of decompositions through the displayed tuple (at $z=1$), with the
distinct rational replacement at $z=2$. Its identity is checked over
$\mathbb Q(z)$ by \texttt{verify\_rank12.py --independent}, separately
from the direct fixed-replacement check in Proposition~\ref{prop:strongsharp}.
Thus this twelve-term decomposition is nonisolated even after quotienting
by permutations and factor gauges; $z$ does not vary the original family
parameter $t$.

\section{Generic rank rigidity through twelve on the parameter curve}\label{app:curve12}

This section concerns only the Li--Wang--Hu parameter curve.  The rank-twelve
conclusion below is not an ambient-component assertion.

\begin{lemma}[Removal of block weights]\label{lem:blockweights12}
Suppose a displayed decomposition is partitioned into \(F,B,U\).  In its
second mode let the factors of \(F\) lie in \(B_F\), and those of \(B,U\)
in \(B_B\), where \(B_F\oplus B_B\) is the factor support.  In its third
mode let the factors of \(F,B\) lie in \(C_0\), and those of \(U\) lie in
\(C_1\), where \(C_0\oplus C_1\) is the factor support.  For
\(f,b,u\in\bbC^\times\), the weighted decomposition
\[
 f\sum_{i\in F}d_i+b\sum_{i\in B}d_i+u\sum_{i\in U}d_i
\]
is carried termwise to the unweighted one by invertible maps in these two
modes.  In particular, its rank and its uniqueness properties are unchanged.
\end{lemma}
\begin{proof}
Scale \(B_F,B_B\) by \(f^{-1},b^{-1}\), respectively, and \(C_0,C_1\)
by \(1,b/u\).  The coefficients of an \(F\)-, \(B\)-, and \(U\)-term
then become \(ff^{-1}\), \(bb^{-1}\), and \(ub^{-1}(b/u)\), respectively.
Extend the maps invertibly to ambient complements.
\end{proof}

Put
\[
 G_1=\{4,17,34,40\},\quad G_2=\{9,25,28,35\},\quad
 G_3=\{12,19,22,32\},\quad G_4=\{15,38,44,48\}.
\]
The Laurent normal form of Lemma~\ref{lem:coreweight}, with
\(\eta=t/2\), sends a core summand to its \(t=2\) value, with the sole
weight \(\eta\) on \(G_4\).  In the four exceptional twelve-support
charts of Section~\ref{app:rank12}, the verified block data are
\[
\begin{array}{c|c|c|c}
 \text{support}&\text{mode order}&(F,B,U)&(f,b,u)\\\hline
 G_1\cup G_2\cup G_3&(A,B,C)&(G_2,G_3,G_1)&(1,1,1)\\
 G_1\cup G_2\cup G_4&(A,C,B)&(G_2,G_4,G_1)&(1,\eta,1)\\
 G_1\cup G_3\cup G_4&(A,C,B)&(G_3,G_1,G_4)&(1,1,\eta)\\
 G_2\cup G_3\cup G_4&(A,B,C)&(G_3,G_2,G_4)&(1,1,\eta).
\end{array}
\]
The symbolic four-chart check verifies the block memberships, including the
zero arm strips, complementary third-mode blocks, nonzero \(\alpha\), and
the scalar cancellations in all twelve terms of every row.  Thus
Lemma~\ref{lem:blockweights12} transports the pointwise completion bound:
each of these four partial sums has rank \(12\) for every \(t\ne0\).
This is an explicit transport argument, not a semicontinuity claim.

\begin{proposition}[Uniform equal-length upper certificate]\label{prop:sharpcurve}
For every \(t\in\bbC^\times\), the full decomposition \(D(t)\) has a
distinct \(48\)-term decomposition at tensor-summand distance exactly \(12\).
Consequently \(\srho(D(t))\le11\).
\end{proposition}
\begin{proof}
The replacement in Proposition~\ref{prop:strongsharp} removes
\(S_0=G_1\cup G_2\cup G_3\), which is disjoint from the weighted tetrad
\(G_4\).  The Laurent normal form therefore carries the removed terms and
the twelve replacement terms termwise to their checked \(t=2\) values.
Apply its inverse to the replacement identity and retain the other \(36\)
displayed terms of \(D(t)\).

The exported certificate checks all \(4096\) tensor coordinates at \(t=2\):
the twelve new terms are nonzero, pairwise distinct, and do not match any
removed term.  Invertibility gives the same nonmatching statement for every
\(t\ne0\).  The \(36\) retained original terms already exhaust the common
submultiset.  If a new term coincided with one of them, it would only add a
second copy and could not increase its common multiplicity; it cannot match
a removed term.  Hence the common tensor-summand multiset has size exactly
\(36\), proving distance \(12\).  This is a statement about complete tensor
summands, not a claim that twelve input-product directions must change.
\end{proof}

\begin{corollary}[Rank twelve and exact strong eleven on the curve]
\label{cor:rank12curve}
There is a nonempty Zariski-open subset \(V\subset\bbC^\times\) containing
\(2\) such that, for every \(t\in V\),
\[
 \rrho(D(t))\ge\KcertRC,\qquad \srho(D(t))=11,\qquad
 \brho(D(t))\ge8.
\]
Consequently every shorter complex decomposition changes at least thirteen
displayed tensor summands of \(D(t)\), and the least distance to a distinct
\(48\)-term decomposition is exactly twelve.
\end{corollary}
\begin{proof}
Start with the open subset supplied by Corollary~\ref{cor:strong11curve};
there every displayed partial sum of size at most eleven is rank-minimal.
Retain also a paired-product minor nonzero at \(2\), so the displayed terms
remain independent.

For the \(1836\) twelve-supports carrying a projected-Koszul certificate,
fix the linear projections used at \(2\) as maps on the original factor
spaces.  The specified \(67\)-minors are Laurent polynomials in \(t\), and
their simultaneous nonvanishing is an open condition containing \(2\).
Each product contributes rank at most six to the map, so a nonzero minor
gives border rank, hence rank, at least \(\lceil67/6\rceil=12\).

Every remaining mixed twelve-support has certified factor-span sum at least
\(25\) at \(2\).  Retain finite factor minors witnessing these bounds.
Together with the preceding proper-face minimality and paired independence,
the rank clause of Lemma~\ref{lem:splitting11} gives rank twelve on their
common open neighborhood.  The only remaining supports are the four unions
of three tetrads, already handled for every nonzero parameter by the
block-weight table above.  These classes are exhaustive: the minors cover
\(1816\) core and \(20\) mixed supports, the table covers the other four
core supports, and the factor-span argument covers every other mixed
support.  Their finite intersection therefore gives \(\rrho(D(t))\ge12\).

Intersect once more with the pullback of the border-eight open set of
Corollary~\ref{cor:generic}.  All opens used contain \(2\).  The strong
lower bound from Corollary~\ref{cor:strong11curve} and the uniform
equal-length certificate of Proposition~\ref{prop:sharpcurve} give
\(\srho(D(t))=11\).  The exact frontier identity gives the shortening
barrier, while the same lower and upper certificates give exact distance
twelve.  Since \(\bbC^\times\) is an irreducible curve, the final open is
nonempty and has finite complement.
\end{proof}

The \(12\) in Corollary~\ref{cor:rank12curve} is curve-specific.  The
ambient-component conclusions remain rank radius at least \(11\), strong
radius at least \(9\), and border radius at least \(8\); no uniform
full-family rank-twelve or border-twelve assertion is made.

\section{Comparison with Koszul--Young certificates}\label{app:kmwcomparison}

This section separates three sufficient certificate classes for a displayed
nonzero complex decomposition $T=\sum_{i=1}^r a_i\otimes b_i\otimes c_i$ with
$r\ge2$.
Let $\mathcal K$ denote the decompositions admitting all eleven hypotheses of
the deterministic Koszul--Young criterion of Kothari--Moitra--Wein
\cite[Thm.~2.7]{KMW}, after a permitted mode order, quotient, and flag. Let
$\mathcal S$ denote the present saturated-projection certificate together with
the finite inequalities \eqref{eq:pairedsylvester}, for some fixed paired
space and finite projection stack. Finally, let $\mathcal R$ denote the
abstract hypotheses of Lemma~\ref{lem:saturatedprojection} together with a
scheme-theoretically reduced paired Segre section, without requiring the
Sylvester inequalities. These are classes of sufficient hypotheses, not
rank-one algorithms or assertions of priority.

\begin{theorem}\label{thm:kmwcomparison}
Over $\bbC$,
\[
 \mathcal K\subsetneq\mathcal R,\qquad
 \mathcal S\subsetneq\mathcal R,\qquad
 \mathcal K\not\subseteq\mathcal S,\qquad
 \mathcal S\not\subseteq\mathcal K.
\]
The strict separations apply after arbitrary invertible changes of factor
bases and nonzero reweighting of the displayed summands. They do not assert
that the full Kothari--Moitra--Wein toolkit cannot prove uniqueness by another
route, nor that this paper introduces Koszul-based rank-one extraction.
\end{theorem}

For the reweighting assertion, a saturated Koszul map factors as $UV$ with
$U$ injective and $V$ surjective. Inserting the invertible diagonal blocks
of nonzero term weights preserves its rank and image; the $N,N',P,P'$ columns
are only rescaled, while the factor dependences and intrinsic-span
obstructions are unchanged. Thus the positive certificates and the all-choice
obstructions used below persist under these reweightings.

\paragraph{Necessary capacity bounds.}
For the Kothari--Moitra--Wein setup, let $q$ be the quotient dimension,
$p$ the exterior degree, and $k=\binom{q-1}{p}$. The two flag matrices
$N,N'$ have $r(k-1)$ columns. If $J$ is a subcollection of $s$ terms, full
column rank requires
\begin{equation}\label{eq:kmwcapacity}
\begin{aligned}
s(k-1)&\le\left[\binom qp-(p+1)\right]\dim\spn\{b_i:i\in J\},\\
s(k-1)&\le\left[\binom q{p+1}-(q-p)\right]\dim\spn\{c_i:i\in J\}.
\end{aligned}
\end{equation}
Indeed, restrict the columns to $J$. They lie in the indicated exterior
coordinate quotient tensored with the intrinsic factor span, and every
subcollection of columns of a full-column-rank matrix is independent.

There is also a quotient-independent obstruction. If a Koszul map is
saturated to $r\binom{d-1}{e}$ in a $d$-dimensional first-mode target, and
the two other intrinsic spans have dimensions $b_0,c_0$, its row and column
dimensions give
\[
 r\binom{d-1}{e}\le\binom de b_0,\qquad
 r\binom{d-1}{e}\le\binom d{e+1}c_0.
\]
Equivalently $r(d-e)\le db_0$ and $r(e+1)\le dc_0$; adding proves
\begin{equation}\label{eq:kmwglobalcapacity}
 r(d+1)\le d(b_0+c_0),\qquad r<b_0+c_0.
\end{equation}
Both bounds survive coordinate changes, linear projections, and finite
ambient enlargements, because they use intrinsic spans. A further elementary
obstruction is that saturation forbids a $C$-factor dependence on at most
$p+1$ terms and a $B$-factor dependence on at most $q-p$ terms: wedge a
decomposable $(p+1)$-vector through the projected first factors in the first
case and use the exterior-dual map in the second. The resulting nontrivial
relation between summand image spaces contradicts saturation.

\paragraph{The certified manuscript exceptions.}
The exact obstruction records for every one of the $32$ exceptional
ten-supports give a four-term dependence and factor-span dimension at most
seven in every mode. The preceding dependence bound leaves only
$(q,p)=(3,1),(4,1),(4,2),(5,2)$. Formula~\eqref{eq:kmwcapacity} then gives
the following all-choice contradiction:
\[
\begin{array}{c|c|c}
 (q,p)&r(k-1)&\text{upper bound for one flag capacity}\\\hline
 (3,1)&10&7\\
 (4,1),(4,2)&20&14\\
 (5,2)&50&49
\end{array}
\]
Thus none of these supports can satisfy the full Kothari--Moitra--Wein
criterion, in any mode order. Write the core as
\[
 G_1=\{4,17,34,40\},\quad G_2=\{9,25,28,35\},\quad
 G_3=\{12,19,22,32\},\quad G_4=\{15,38,44,48\}.
\]
The same records show that every eleven-subset of this core has an
$A$-dependence of size at most six and four-term dependencies in both $B$ and
$C$. Thus the three mode caps are $(5,3,3)$. Put $s=p+1$ and $t=q-p$.
Across all mode orders, the possible unordered pairs and the weaker global
flag capacity are
\[
\begin{array}{c|rrrrrrr}
 \{s,t\}&\{2,2\}&\{2,3\}&\{2,4\}&\{2,5\}&\{3,3\}&\{3,4\}&\{3,5\}\\\hline
 11(k-1)&11&22&33&44&55&99&154\\
 \text{capacity}&8&16&24&32&56&96&144.
\end{array}
\]
Only $s=t=3$, namely $(q,p)=(5,2)$, survives. In mode $B$ each of
$G_1\cup G_3,G_2\cup G_4$ spans at most four dimensions, and in mode $C$
the corresponding low-span blocks are $G_1\cup G_4,G_2\cup G_3$. Whichever
matrix mode is used, an eleven-subset meets one such eight-block in at least
six terms; \eqref{eq:kmwcapacity} would require $30\le28$.
Hence all $\binom{16}{11}=4368$ core eleven-supports are excluded, including
the $656$ exceptional supports of Section~\ref{app:strong11}. These are
dependence-and-dimension proofs, not failures of a selected numerical flag.

Conversely, the positive certificates in this paper establish the
Sylvester-equipped class for the sixteen saturated ten-cases and the $624$
saturated eleven-cases: their Koszul maps are saturated, their stacked
original paired kernels have the displayed dimension, and
\eqref{eq:pairedsylvester} excludes every other complex rank-one direction.
The completion cases are not reclassified as saturated cases. Thus these
$640$ supports give $\mathcal S\setminus\mathcal K$ witnesses. Notice that
the inequality \eqref{eq:pairedsylvester} is also condition $(H_2)$ of
Domanov--De Lathauwer \cite{DDL}; this identifies a framework overlap, not
a subsumption of the stacked-kernel construction. The full paired quadratic
rank-one extraction test also succeeds on these supports, so the separation
is not a claim that Sylvester exclusion is generally stronger than polarized
minor extraction.

\paragraph{A reverse certificate.}
An exact integer certificate, documented in the mathematical supplement,
defines an eighteen-term decomposition in format $5\times13\times13$.
For $(q,p)=(5,2)$, explicit nonzero minors at the two primes $1000003$ and
$65521$, together with direct nonvanishing and independence checks, certify
all eleven Kothari--Moitra--Wein conditions; the six matrix ranks are
\[
\begin{array}{c|c|c|c}
\text{matrices}&\text{dimensions}&\text{certified rank}&\text{required rank}\\\hline
M,M'&130\mathbin{\times}130&108&18\binom42=108\\
N,N'&91\mathbin{\times}90&90&90\\
P,P'&234\mathbin{\times}153&153&\binom{18}{2}=153.
\end{array}
\]
These positive integer-minor certificates lift to $\bbQ$ and $\bbC$.
If the five-dimensional factor is distinguished, either paired factor-span
sum is at most $5+13$, so the full eighteen-support has Sylvester left side
$5+13-18=0<2$. If either thirteen-dimensional factor is distinguished,
\eqref{eq:kmwglobalcapacity} would require $18<5+13$, which is false.
This exhausts mode orders, projection stacks, and finite enlargements, so the
displayed decomposition lies in $\mathcal K\setminus\mathcal S$.

\paragraph{From KMW to a reduced paired kernel.}
We give the finite lifting argument establishing the two containments in
Theorem~\ref{thm:kmwcomparison}. Let $\Pi:A\to Q$ be a KMW quotient, put
$u_i=\Pi a_i$, $h=q-p$, and let $f_i$ be the first $h$ coordinates of $u_i$.
For the $A\otimes C$ paired kernel, use the primed KMW map:
exterior Hodge duality identifies $M'$ with a signed row/column permutation
of $M^{\mathsf T}$, but its flag section is the one whose paired factors are
the $C$ factors. On this primed flag side, KMW's saturated image and flag
intersection conditions imply that the common projected paired kernel is
\[ S_Q=\spn\{u_i\otimes c_i:1\le i\le r\}. \]
For completeness, the designated primed column indexed by $[h]$ puts the
first $h$ rows of any kernel element in
$S_f=\spn\{f_i\otimes c_i\}$. Subtract that unique displayed combination.
For $j>h$, the swapped column
$U=([h]\setminus\{2\})\cup\{j\}$ and the flag identification give
$e_2\otimes Z_j\in S_f$, where $Z_j$ is the remaining $j$th row. This is a
rank-one point with zero first coordinate. KMW's polarized-minor matrix $P'$ has full column
rank, so its restricted minors span every cross monomial $z_i z_j$ in the
displayed coefficient coordinates. Their ideal is the coordinate-axis ideal,
and its projective scheme is exactly the $r$ reduced displayed points, all
with nonzero first coordinate; hence every $Z_j$ vanishes. Every other
restricted minor vanishes on the coordinate axes, hence has no square term
and lies in the cross-monomial ideal. Thus the whole restricted ideal equals
that ideal, proving equality and reducedness without a point count. This is the linear-section
extraction used in KMW and related to the framework of
Johnston--Lovitz--Vijayaraghavan \cite{JLV}; it is stated here only for the
certificate implication above.

It remains to lift from $Q$ if $\Pi$ has a kernel. Write
$A=A_0\oplus D$, $D=\ker\Pi$, and choose $v=e_2\in Q$, which is nonzero and
has first coordinate zero. The displayed projected factors have nonzero first
coordinate, so $S_Q\cap(v\otimes C)=0$. For a basis $e_j$ of $D$, let
$\ell_j$ be the coordinate functional vanishing on $A_0$, and replace the
quotient by $\Pi_j(t)=\Pi+t v\ell_j$. The finitely many KMW minors remain
nonzero for some nonzero $t$, so each $\Pi_j(t)$ has the preceding projected
kernel description. An element in the stack kernel, after subtracting its
displayed paired part, is $R=\sum_j e_j\otimes w_j$. The $j$th perturbed
kernel identity gives
\[
 tv\otimes w_j=\sum_i\mu_i
   (u_i+t\ell_j(a_i)v)\otimes c_i
\]
for some $\mu_i$. Rearrangement puts
\[
 \sum_i\mu_i u_i\otimes c_i\in S_Q\cap(v\otimes C)=0.
\]
Independence of the displayed pairs gives $\mu_i=0$ for every $i$, and then
$w_j=0$ because $t$ and $v$ are nonzero. The stack kernel is therefore the
original displayed paired span. The prefix minors remain cross-monomial
generators after pullback, so its Segre section is reduced as well. This
uses at most $1+\dim\ker\Pi$ projections. Only after this lifting do we
apply Lemma~\ref{lem:projectionambient} to certificates in larger spaces.
Lemma~\ref{lem:saturatedprojection}
then gives $\mathcal K\subseteq\mathcal R$; the ten-support witnesses make
the inclusion strict. Lemma~\ref{lem:pairedreduced} gives
$\mathcal S\subseteq\mathcal R$, and the eighteen-term witness makes that
inclusion strict.

The comparison credits the existing Koszul--Young, linear-section, and
subset-rank tools. It neither establishes priority for finite lifting,
multi-projection stacking, reduced-incidence transfer, or zero-shear purity,
nor claims an ambient generic strong-eleven theorem or a new lower bound for
matrix multiplication.

\section{Consequences}\label{sec:consequences}

\begin{corollary}\label{cor:mainbarrier}
Every decomposition of $\M{4,4,4}$ of length less than $48$ over $\bbC$ differs from $D(2)$ in at least $\Knext$
% NUMERAL: 35 = 48 - \Knext
summands; a length-$47$ decomposition, if one exists, reuses at most $35$ of $D(2)$'s $48$ summands.
\end{corollary}

\begin{proof}
Theorem~\ref{thm:frontier} with the rank radius of Theorem~\ref{thm:rank12}.
\end{proof}

\begin{corollary}\label{cor:other48s}
Any $48$-term decomposition of $\M{4,4,4}$ over $\bbC$ that is not equal to $D(2)$ up to permutation and gauge
differs from it in at least $\KcertSNext$ summands, and this bound is attained
by Proposition~\ref{prop:strongsharp}.
\end{corollary}

\begin{proof}
Cancel the common summand multiset. If fewer than twelve summands changed,
the two remaining decompositions would contradict Theorem~\ref{thm:strong11}.
Proposition~\ref{prop:strongsharp} attains the bound.
\end{proof}

\begin{corollary}\label{cor:curve}
Any decomposition of $\M{4,4,4}$ over $\bbC$ obtained from $D(2)$ by changing $k$ of its $48$ complete summands has
% NUMERAL: 72-2k = 48 - 2(k - \KcertR)
length at least $48-2\max(0,k-\KcertR)$, which for $k\ge\KcertR$ reads $72-2k$ and for $k\le\KcertR$ says only
% NUMERAL: 16 and 20 solve 72-2k = 40 and 32
that the length is at least $48$. Reaching length $40$ therefore requires changing at least $16$ summands, and
reaching length $32$ at least $20$.
\end{corollary}

\begin{proof}
Corollary~\ref{cor:increment} with $\rrho(D(2))\ge\KcertR$ from Theorem~\ref{thm:rank12}.
\end{proof}

If a shorter decomposition exists, Theorem~\ref{thm:exactdist} gives
$d_<(D(2))=\rrho(D(2))+1$; otherwise $d_<(D(2))=\infty$. In either case no single rewrite touching at most
$\KcertR$ summands shortens $D(2)$; and by Theorem~\ref{thm:geodesic} every path in the rewrite graph to a
shorter decomposition has endpoints differing in at least $\Knext$ complete summands,
including paths through longer decompositions, which is
what adaptive flip search \cite{AIH} traverses. By Corollary~\ref{cor:other48s} even moving to a different
scheme of the same length requires changing at least $\KcertSNext$, and at least four at every other
parameter by Corollary~\ref{cor:other48}. The distance $d(D,D')$ between two
fixed displayed decompositions can change under an isotropy applied to only
one argument. By contrast, $d_<(D)$ minimizes over all shorter
competitors and is an isotropy invariant: for an isotropy $g$,
$d_<(gD)=d_<(D)$, since $D'\mapsto gD'$ is a length-preserving
bijection of the competitor set and $d(gD,gD')=d(D,D')$.
Thus the shortening minimum $d_<(D)$, and hence its lower bound, depends only
on the isotropy orbit of the base decomposition. Simultaneously applying an
isotropy to $D(2)$ and to its equal-length competitors likewise transports
Corollary~\ref{cor:other48s} to every representative; this does not assert
invariance of the distance between two fixed representatives when only one is
transformed.

\section{Data, code, and reproducibility}\label{sec:repro-contract}

The accompanying mathematical supplement records the source layout, exact
factor data, certificate coverage, environments, and replay commands.  The
full profile checks the pointwise ten/eleven/twelve chain, the sharp
replacement, uniform core, border-eight certificates, and the KMW comparison.
Each certificate is accepted only with its prerequisite and coverage checks;
unresolved searches are not treated as proofs.  Historical timing logs,
alternate transfers, and application censuses are retained in the supplement
because they are useful audit material but not premises of the main results.

The frozen version-1 evidence bundle is publicly available at
\url{https://doi.org/10.5281/zenodo.22684265}\cite{AgarwalZenodo2026}.
It is the file \path{rigidity-jsc-revised-evidence.zip} with SHA-256
\texttt{e3663c16a2fd2b17bd9df9eab295e63330b75c90ea8fab9bca4584bbac915e88};
it contains the editable sources, both PDFs, exact input data, and a fail-closed
replay script. This article's later citation-only PDF/source revisions are not
claimed to be members of that immutable deposited bundle.
Its README maps the certificate paths in the supplement into the
archived proof dependency. The new four-chart transport and complete-summand
distance checks are under \path{evidence/curve_twelve/} and
\path{evidence/metrics/}, respectively. Each invocation starts a fresh result
directory; the ledger distinguishes freshly replayed certificates from
retained application logs. The archive's checksum manifest identifies the
files used, but is not itself a mathematical certificate.

\section*{Funding and competing interests}
The author is an independent researcher and received no external funding for
this work. The author declares no competing interests.

\section*{Declaration of generative AI and AI-assisted technologies in the manuscript preparation process}
During the preparation of this work, the author used Codex for verification
and tooling. The author reviewed and edited the output as needed and takes
full responsibility for the content of the published article.
For a fuller research-method account, see Appendix~B of the mathematical
supplement.

\end{document}